\documentclass[10pt,a4paper]{article}
\usepackage[no-natbib-sort]{jheppub}
\usepackage[utf8]{inputenc}
\usepackage[english]{babel}
\usepackage{amsmath}
\usepackage{amsfonts}
\usepackage{amssymb}
\usepackage{amsthm}

\newcommand{\bb}[1]{\boldsymbol{#1}}
\newcommand{\bbu}[1]{\underline{\boldsymbol{#1}}}
\newcommand{\bsigma}{\bb\sigma}
\newcommand{\bxi}{\bb\xi}
\newcommand{\btau}{\bb\tau}
\newcommand{\sumsigma}{\sum_{\bb\sigma\in\Sigma_N}}
\newcommand{\sumsigmap}{\sum_{\bb\sigma'\in\Sigma_N}}
\newcommand{\mR}{\mathbb R}
\newcommand{\vx}{\bb x}

\newcommand{\vv}{\bb v}
\newcommand{\average}[2]{\omega_{#1}(#2)}

\newcommand{\E}{\mathbb{E}}

\newcommand{\qavJ}{\mathbb E_{\bb J}}

\theoremstyle{plain}
\newtheorem{Rmk}{Remark}
\newtheorem{Thm}{Theorem}
\newtheorem{Lem}{Lemma}
\newtheorem{Prp}{Proposition}
\newtheorem{Cor}{Corollary}
\newtheorem{Def}{Definition}

\title{Non-linear PDEs approach to statistical mechanics of Dense Associative Memories}
\author[a,b]{Elena Agliari,}
\author[a,b]{Alberto Fachechi,}
\author[a,b]{Chiara Marullo}
\affiliation[a]{Dipartimento di Matematica “Guido Castelnuovo”, Sapienza Università di Roma, Roma, Italy}
\affiliation[b]{GNFM-INdAM, Gruppo Nazionale di Fisica Matematica, Istituto Nazionale di Alta Matematica,
	Lecce, Italy}
\emailAdd{agliari@mat.uniroma1.it}
\abstract{
Dense associative memories (DAM), are widespread models in artificial intelligence used for pattern recognition tasks; computationally, they have been proven to be robust against adversarial input and theoretically, leveraging their analogy with spin-glass systems, they are usually treated by means of statistical-mechanics tools.
Here we develop analytical methods, based on nonlinear PDEs, to investigate their functioning. In particular, we prove differential identities involving DAM’s partition function and macroscopic observables useful for a qualitative and quantitative analysis of the system. These results allow for a deeper comprehension of the mechanisms underlying DAMs and provide interdisciplinary tools for their study.
}

\begin{document}
\maketitle

\section{Introduction} \label{sec:intro}
{Artificial intelligence (AI) is rapidly changing the face of our society thanks to its impressive abilities in accomplishing complex tasks and extracting information from non-trivially structured, high-dimensional datasets. The successful applications of modern AI range from hard sciences and technology to more practical scenarios ({\it e.g.} medical sciences, economics and finance, and many daily tasks). Its important success is primarily due to the ascent of deep learning \cite{LeCun2015,Schmidhuber2015}, a set of semi-heuristic techniques consisting in stacking together several minimal building blocks in complex architectures with extremely-high learning performances. Despite its success, a rigorous theoretical framework guiding the development of such machine-learning architectures is still lacking. In this context, statistical mechanics of complex systems offers ideal tools to study neural network models from a more theoretical (and rigorous) point of view, thus drawing a feasible path which makes AI less empirical and more explainable.

In statistical mechanics, a primary classification of physical systems is the following. On the one side, we have simple systems, which are essentially characterized by the fact that the number of equilibrium configurations does not depend on the system size $N$. A paradigmatic (mean-field) realization of this situation is the Curie-Weiss (CW) model, in which all the spins $\sigma_i$, $i=1,\dots,N$, making up the system interact pairwisely by a constant, positive ({\it i.e.} ferromagnetic) coupling $J$. Below the critical temperature, in fact, the system exhibits ordered collective behaviors, and the equilibrium configurations of the system are characterized by only two possible values of the global magnetization $m(\bb\sigma):=\frac1N\sum_{i=1}^N\sigma_i$ (which are further related by a spin-flip symmetry $\sigma_i \to -\sigma_i $ for each $i=1,\dots,N$). On the opposite side, we have complex systems, in which the number of equilibrium configurations increases according to an appropriate function of the system size $N$ due to the presence of frustrated interactions \cite{mezard1988spin}. The prototypical example of mean-field spin-glass is the Sherrington-Kirkpatrick (SK) model \cite{sherrington1975}, in which the interaction strengths between the spin pairs are i.i.d. Gaussian variables. With respect to simple systems, spin-glass models exhibit a richer physical and mathematical structure, as shown by the presence of the spontaneous replica-symmetry breaking and an infinite number of phase transitions (\emph{e.g.} see \cite{Parisi1979rsb1,Parisi1979rsb2,Parisi1980rsb3,Parisi1980sequence,mezard1984replica,guerra2003broken,ghirlanda1998general,talagrand2000rsb}) as well as the ultrametric organization of pure states  (\emph{e.g.} see \cite{panchenko2010ultra,panchenko2011ultra,panchenko2013parisi}). Statistical mechanics of spin glasses has acquired a prominent role during the last decades due to its ability to describe the equilibrium dynamics of several paradigmatic models for AI, in particular thanks to the work by Amit, Gutfreund, and Sompolinsky \cite{amit1985} for associative neural networks. For our concerns, the relevant ones are the Hopfield model \cite{hopfield1982hopfield,pastur1977exactly} and its $p$-spin extensions, the Dense Associative Memories (DAMs) \cite{Krotov2016DenseAM,Krotov2018DAMS,AD-EPJ2020}, exhibiting features which are peculiar both of ferromagnetic (simple) and spin-glass (complex) systems. In these models, the interactions between the spins are designed in order to store $K$ ``information patterns'', each encoded by $K$ binary vectors of length $N$ and denoted by $\{\bb \xi^\mu\}_{\mu=1,\dots,K}$, with $\boldsymbol \xi^{\mu}=(\xi_1^{\mu},\xi_2^{\mu},...,\xi_N^{\mu}) \in \{-1, +1\}^N$ and $\xi^\mu_i$ is a Rademacher random variable for any $i=1,\dots,N$ and $\mu=1,\dots,K$; the $\mu$-th pattern is said to be stored if the configuration $\boldsymbol \sigma = \boldsymbol \xi^{\mu}$ is an equilibrium state and the relaxation to this configuration, starting from a relatively close one (\emph{i.e.} a corrupted version of $\boldsymbol \xi^{\mu}$), is interpreted as the retrieval of that pattern. The Hamilton function (or the energy in physics jargon) of these systems can be expressed as
$$
H_{N,p}(\bb\sigma)\propto -\sum_{\mu=1}^K( m_{\mu }(\bb\sigma))^p,
$$
where $p$ is the interaction order (for the Hopfield model $p=2$, while $p>2$ for the DAMs) and $m_\mu (\bb\sigma):=\frac1N\sum_{i=1}^N \xi^\mu_i \sigma_i$ is the so-called Mattis magnetizations measuring the retrieval of the $\mu$-th pattern. It has been shown that the number of storable patterns scales, at most, as a function of the system size, more precisely, $K < \alpha_c(p) N^{p-1}$, where $\alpha_c(p) \in \mathbb R$ depends on the interaction order $p$ and is referred to as critical storage capacity \cite{baldi1987number,bovier2001spin}. By a statistical-mechanics investigation of these models one can highlight the macroscopic observables (order parameters) useful to describe the overall behavior of the system, namely to assess whether it exhibits retrieval capabilities, and the natural control parameters whose tuning can qualitatively change the system behavior; such knowledge can then be summarized in \emph{phase diagram}. 

Regarding the methods, statistical mechanics offers a wide set of techniques for analyzing the equilibrium dynamics of complex systems, and in particular to solve for their free-energy. Historically, the first method (which was applied to the SK model and the Hopfield model \cite{amit1985,steffan1994replica}) is the {\it replica trick}, which -- despite being straightforward and effective -- is semi-heuristic and suffers from delicate points, see for example \cite{tanaka2007moment}. Alternative, rigorous approaches were developed during the years and, among these, the relevant one for our concerns is Guerra's interpolating framework. In this case, we can take advantage of rigorous mathematical methods by applying sum rules \cite{guerra2001sum} or by mapping the relevant quantities (the free energy or the model order parameters) of the statistical setting to the solutions of PDE systems. Indeed, differential equations involving the partition functions (or related quantities) of thermodynamic models have been extensively investigated in the literature, see for example \cite{agliari2012notes,barra2013mean, barra2008mean, barra2010replica,barra2014proc,Moro2018annals,AMT-JSP2018, Moro2019PRE, ABN-JMP2019, fachechi2021pde,AAAF-JPA2021}. In particular, they allow us to express the equation of state (or the self-consistency equations) governing the equilibrium dynamics of the system in terms of solutions of non-linear differential equations, and to describe phase transition phenomena as the development of shock waves, thus linking critical behaviours to gradient catastrophe theory \cite{Barra2015Annals,DeNittis2012PrsA,Giglio2016Physica,Moro2014Annals}. In a recent work \cite{fachechi2021pde}, a direct connection between the thermodynamics of ferromagnetic models with interactions of order $p$ and the equations of the Burgers hierarchy was established by linking the solution of the latter as the equilibrium solution of the order parameter of the former ({\it i.e.} the global magnetization $m$). In the present paper, we extend these results to complex models, in particular to the Hopfield model and the DAMs.
\par\medskip
The paper is organised as follows. In Section 2, we introduce the relevant tools for our investigations, in particular Guerra's interpolating scheme for the PDE duality. In Section 3, as a warm up, we review some basic results about the $p$-spin ferromagnetic models. 
In Section 4, we extend our results to the Derrida models (constituting the $p-$spin extension of the SK spin glass) \cite{derrida1981rem}.
In Section 5, we merge our results in a unified methodology for dealing with the DAMs, especially in the so-called high storage limit, and re-derive the self-consistency equations for the order parameters by means of PDE technology.
}

\section{Generalities and notation}
In this Section, we present the thermodynamic objects we aim to study. We start with a system made up of $N$ spins whose configurations $\bb \sigma\in\Sigma_N  \equiv \{-1,+1\}^N$ are the nodes of a hypercube and interacting via a suitable tensor $\bb J$ of order $p$. The Hamilton functions of the system we will consider in this paper are of the form
\begin{equation}
	\label{eq:Gen_Ham}
	H_{N,p,\bb J}(\bb \sigma)= -\frac1{D_{p,N,J}} \sum_{i_1,\dots,i_p=1}^{N}J_{i_1,i_2,\dots,i_p}\sigma_{i_1}\sigma_{i_2}\dots \sigma_{i_p},
\end{equation}
where $D_{p,N,J}$ is a normalization factor ensuring the linear extensivity of the energy with the system size.
Once the Hamiltonian is fixed, we introduce the partition function in the usual Boltzmann-Gibbs form. Thus, given $\beta \in \mR_+$ the level of thermal noise of the system, the partition function is defined as
\begin{equation}\label{eq:Gen_part_func}
	Z_{N,p,\bb J }(\beta ):= \sum_{\bb\sigma\in\Sigma_N} \exp\left[-\beta H_{N,p,\bb J }(\bb\sigma)\right].
\end{equation}
For simple systems, the partition function can be computed exactly for any $\bb J$ coupling matrix. As is standard in statistical mechanics, it is convenient to compute intensive quantities which are well-defined in the thermodynamic limit $N\to\infty$. Since the partition function is a sum of $2^N$ contribution, it is sufficient to take the intensive logarithm of the partition function, {\it i.e.}
$$
A_{N,p,\bb J}(\beta):=\frac1N \log Z_{N,p,\bb J }(\beta ),
$$
which is the intensive statistical pressure (which is, apart for a factor $-\beta$, the usual free energy) of the system.
However, when dealing with spin-glass systems, the coupling tensor $\bb J$ is a multidimensional random variable, thus the partition function defines a random measure on the configuration space. For good enough probability distributions of the coupling matrix, the intensive logarithm of the partition function is expected to converge to its expectation value in the thermodynamic limit $N\to\infty$ by virtue of self-averaging theorems \cite{ShcherbinaPastur-JSP1991,Bovier-JPA1994}, so it is natural to consider the quenched intensive pressure associated to the partition function \eqref{eq:Gen_part_func}, which is defined as
 \begin{equation}
 	A_{N,p}(\beta):=\frac 1 N\E_{\bb J}\log Z_{N,p, \bb J }(\beta),
 \end{equation}
 where $\E_{\bb J}$ denotes the average over the quenched disorder $\bb J$ (we stress that, in this case, the free energy does not depend any longer on the coupling matrix because of the average operation).\par\medskip
Rather than working with the quantity \eqref{eq:Gen_part_func}, we will use as a fundamental object Guerra's interpolated partition function and its associated interpolating intensive pressure. For instance, for spin-glass systems we would have
\begin{equation}
\begin{split}
	\label{eq:Gen_Guerra}
Z_{N,p, \bb J }(t,\vx)&:=\sum_{\bb \sigma\in\Sigma_N} \exp\left[-{H_{N,p, \bb J }(t,\vx)}\right],
\\  A_{N,p}(t,\vx)&:=\frac 1 N \E_{\bb J } \log Z_{N,p, \bb J }(t,\vx),
\end{split}
\end{equation}
where $H_{N,p, \bb J }(t,\vx)$ denotes the interpolating Hamiltonian satisfying the properties that, at $\bb x=0$ and $t \neq0$, it recovers the Hamiltonian \eqref{eq:Gen_Ham} times $\beta$, and at $t=0$ and $\bb x \neq \bb 0$ it corresponds to an exactly-solvable $1$-body system, namely a system where spins interact only with an external field that has to be set a posteriori. 
The interpolating parameters $t$ and $\bb x$ are interpreted, in a mechanical analogy, as spacetime coordinates with suitable dimensionality. 
\par\medskip
The interpolating structure \eqref{eq:Gen_Guerra} implies a generalized measure, whose related Boltzmann factor is 
\begin{equation}
B_{N,p, \bb J }(t,\vx):=\exp\left[-{H_{N,p, \bb J }(t,\vx)}\right].
\end{equation}
Thus, for an arbitrary observable $ O(\bb \sigma)$ in the configuration space $\Sigma_N$, we can introduce the Boltzmann average induced by the partition function \eqref{eq:Gen_Guerra} as
\begin{equation}\label{B_average}
\omega_{t,\vx}( O):= \frac{1}{Z_{N,p, \bb J }(t,\vx)}\sum_{\bb \sigma\in\Sigma_N} O(\bb \sigma)B_{N,p, \bb J }(t,\vx).
\end{equation}
Usually, in spin-glass systems, the quenched average is performed after taking the Boltzmann expectation values on the $s$-replicated space $\Sigma_N^{(s)}=( \Sigma_N)^{\otimes s}\equiv \{-1,+1\}^{sN}$, which is naturally endowed with a random Gibbs measure corresponding to the partition function $Z_{N,p, \bb J }^{(s)}(t,\vx)=Z_{N,p, \bb J }(t,\vx)^s$. Given a function $O:\Sigma_N^{(s)}\to \mR$, the Boltzmann average in the $s$-replicated space are straightforwardly defined as
$$
\Omega_{t,\vx}^{(s)}(O):= \frac{1}{Z_{N,p, \bb J }^{(s)}(t,\vx)}\sum_{\underline{\bb \sigma}\in\Sigma_N^{(s)}} O(\underline{\bb \sigma})B_{N,p, \bb J }^{(s)}(t,\vx)
$$
where $\underline{\bb \sigma} \in \Sigma_N^{(s)}$ is the global conﬁguration of the replicated system, and $B_{N,p, \bb J }^{(s)}(t,\vx)$ is the Boltzmann factor associated to the $s$-replicated partition function. Of course, in spin-glass theory, the relevant quantities are the {\it quenched} expectation values, which are defined as
\begin{equation}
\langle  O\rangle_{t,\vx}:=\E_{\bb J}\Omega_{t,\vx}^{(s)}(O).
\end{equation}
For the sake of simplicity, we dropped the index $s$ from the quenched averages, as it would be clear from the context.\par\medskip

With all these definitions in mind, we are then able to find the link between the resolution of the statistical mechanics of a given spin-like model and a specific PDE problem in the fictitious space $(t,\vx)$. Before concluding this section it is worth recalling that here we will work under the replica-symmetry (RS) assumption, meaning that we assume the self-averaging property for any order parameter $X$, {\it i.e.} the fluctuations around their expectation values vanish in the thermodynamic limit. In distributional sense, this corresponds to
\begin{equation}
\lim_{N\to\infty} \mathcal P_{t,\vx} (X)=\delta(X-\bar X).
\end{equation}
where $\bar X= \langle X\rangle _{t,\vx}$ is the expectation value w.r.t. the interpolating measure $ \mathcal P_{t,\vx} (X)$. Typically, for simple systems this assumption is correct, conversely, for complex systems this is not always the case, for instance, in spin-glasses the RS is broken at low temperature \cite{mezard1988spin}. When dealing with neural-network models, RS constitutes a standard working assumption as it usually applies (at least) in a limited region of the parameter space, while elsewhere it yields only small quantitative discrepancies with respect to the exact solution \cite{amit1989,SpecialIssue-JPA}. The latter, accounting for RSB phenomena, can be obtained by iteratively perturbing the RS interpolation scheme (\emph{e.g.} see \cite{AABO-JPA2020,AAAF-JPA2021,albanese2021}), thus, our results find direct application on the practical side and provide the starting point for further refinements on the theoretical side.

\section{$p$-spin ferromagnetic models: how to deal with simple systems}
The present section is a compendium of the results reported in \cite{fachechi2021pde}, so we refer to that work for a detailed derivation. In $p$-spin ferromagnets, the interaction between spins is fixed by the requirement $J_{i_1,i_2,\dots,i_p}=J$ for each $i_1,\dots,i_p=1,\dots ,N$ and $J>0$; without loss of generality, one can set $J=1$, since it corresponds to a rescaling of the thermal noise. Thus, the Hamilton function of the model simply reads as
\begin{equation}
\label{eq:CW_original}
H_{N,p}(\bb\sigma):=-\frac{1}{ N^{p-1}}\sum_{i_1,\dots,i_p=1}^N\sigma_{i_1}\dots \sigma_{i_p} = -N (m (\bb\sigma))^p,
\end{equation}
with $$m(\bb\sigma):=\frac 1N \sum_i\sigma_i,$$ being the global magnetization of the system. By following the same lines of \cite{fachechi2021pde}, Guerra's interpolating partition function reads as
\begin{eqnarray}
	Z_{N,p}(t,x)&=&\sum_{\bb \sigma\in\Sigma_N}\exp\left(-H_{N,p}(t,x)\right),\label{Z_int_pspin}\\
	H_{N,p}(t,x)&=&t N m(\bb\sigma)^p-Nx m(\bb\sigma),
\end{eqnarray}
where $(t,x)\in \mR^2$.  The starting point is to notice that the interpolating statistical pressure associated to the partition function \eqref{Z_int_pspin} has spacetime derivatives
\begin{eqnarray}
	\partial_t A_{N,p}(t,x)&=&-\omega_{t,x} (m(\bb\sigma)^p),\\
	\partial _x A_{N,p}(t,x)&=& \omega_{t,x}(m (\bb\sigma)).
\end{eqnarray}
The expectation value of monomials of the global magnetization satisfies the following relation \cite{fachechi2021pde}:
\begin{equation}
\partial_x \omega_{t,x} (m(\bb\sigma)^s)=N (\omega_{t,x}(m (\bb\sigma)^{s+1})-\omega_{t,x} (m(\bb\sigma)^s)\omega_{t,x}(m(\bb\sigma))).
\end{equation}
This means that we can act on the expectation value $\omega_{t,x}(m(\bb\sigma))$ to generate higher momenta. In particular, calling $u(t,x)=\omega_{t,x}(m(\bb\sigma))$ and setting $s=p-1$, we directly get the Burgers hierarchy
\begin{equation}\label{Burger_hierarchy}
\partial_t u(t,x)+\partial_x\left(\frac 1 N \partial_x+u(t,x)\right)^{p-1} u(t,x)=0.
\end{equation}
This duality also allows us to analyse the thermodynamic limit, corresponding to the inviscid scenario for the Burgers hierarchy. Indeed, posing $\bar u(t,x)=\lim_{N\to\infty}u(t,x)=\lim_{N\to\infty}\omega_{t,x}(m(\bb\sigma))$, we have the initial value problem
\begin{equation}\label{Burger_inviscid}
\begin{cases}
\partial_t \bar u(t,x)+p\bar u(t,x)^{p-1} \partial_x\bar u(t,x)=0\\
\bar u(0,x)=\tanh(x)
\end{cases},
\end{equation}
where the initial profile is easily computed by straightforward calculations (since it is a 1-body problem). This system describes the propagation of non-linear waves, and can be solved by assuming a solution in implicit form $\bar u(t,x)=\tanh (x-v(t,x)t)$, where $v(t,x)=p \bar u(t,x)^{p-1}$ is the effective velocity. Recalling that the thermodynamics of the original $p$-spin model associated to the Hamilton function \eqref{eq:CW_original} is recovered by setting $t =-\beta$ and $x = 0$, we directly obtain
\begin{equation}\label{eq:pferr_sc}
\bar m=\tanh(\beta p\bar m^{p-1}),
\end{equation}
where $\bar m=\bar u(-\beta,0)$. This is precisely the self-consistency equation for the global magnetization for the $p$-spin ferromagnetic model \cite{fachechi2021pde}. The phase transition of the system is expected to take place where the gradient of the solution explodes, which, on the Burgers side, corresponds to the development of a shock wave at $x=0$. Since the temporal coordinate $t$ is directly related to the thermal noise at which the phase transition occurs, with standard PDE methods we can analytically determine the critical temperature according to the simple system
$$
\begin{cases}
	\bar \xi =\frac{F(\bar \xi)}{F'(\bar \xi)},\\
	T_c = F'(\bar \xi),
\end{cases}
$$ 
where $F(\xi)=p\tanh (\xi)^{p-1}$. This prediction is in perfect agreement with the numerical solutions of the self-consistency equation \eqref{eq:pferr_sc}.

\section{Derrida models: how to deal with complex systems}
In this section, we adapt the previous methodologies to treat complex systems with $p$-spin interactions. The paradigmatic case is given by the $p$-spin SK model, also referred to as Derrida model, defined as follows 

\begin{Def}
	Let $\bb\sigma $ be the generic point in the configuration space $\Sigma_N=\{-1,+1\}^N$ of the system. Let $\bb J$ be a $p$-rank random tensor with entries $J_{i_1\dots i_p}\sim \mathcal N(0,1)$ i.i.d. The Hamilton function of the $p$-spin Derrida model is defined as
	\begin{equation} \label{eq:def_deridda}
		H_{N,p, \bb J }(\bb\sigma)= -\sqrt{\frac{p!}{2 N^{p-1}}}\sum_{1\le i_1<\dots<i_p \le N}^N J_{i_1 \dots i_p}\sigma_{i_1}\dots \sigma_{i_p}.
	\end{equation}	
\end{Def}
\begin{Rmk}
	Clearly, for $p=2$ we recover the Sherrington-Kirkpatrick model \cite{sherrington1975}.
\end{Rmk}
\begin{Rmk}\label{Rmk:sum}
In the usual definition of the $p$-spin SK model, the sum is performed with the constraint $1\le i_1<i_2<\dots i_p\le N$ like in \eqref{eq:def_deridda}. Beyond that formulation, it is possible to consider an alternative one, where summation is realized independently over all the indices, the difference between the two prescriptions being vanishing in the thermodynamic limit, that is
\begin{equation}
	\label{eq:subleading}
		\sum_{1\le i_1<\dots<i_p \le N}(\cdot)= \frac1{p!}\sum_{i_1,\dots,i_p =1}^N(\cdot) + \textnormal{ contributions vanishing as } N\to\infty.
\end{equation}
Since we are interested in the thermodynamic limit, we will often use the equality
	\begin{equation}
		\sum_{1\le i_1<\dots<i_p \le N}(\cdot)= \frac1{p!}\sum_{i_1,\dots,i_p =1}^N(\cdot),
	\end{equation}
holding in the $N\to\infty$ limit.
\end{Rmk}

\begin{Def} \label{def:DM}
	Given $(t,x)\in \mR^2$ and given a family $\{J_i\}_{i=1}^N$ of i.i.d. $\mathcal N(0,1)$-distributed random variables, Guerra's interpolating partition function for the $p$-spin SK model is
	\begin{eqnarray}\label{ZSK}
		Z_{N,p, \bb J } (t,x)&=&\sum_{\bb \sigma \in \Sigma_N} \exp\big(-H_{N,p, \bb J }(t,x)\big),\\
		H_{N,p, \bb J }(t,x)&=&-\sqrt{\frac{ t p!}{2 N^{p-1}}}\sum_{1\le i_1<\dots<i_p \le N}^N J_{i_1 \dots i_p}\sigma_{i_1}\dots \sigma_{i_p}-\sqrt x \sum_{i=1}^N J_i \sigma_i.
	\end{eqnarray}
The Boltzmann factor associated to this partition function is denoted with $B_{N,p,\bb J} (t,x) $.
\end{Def}
As stated in Sec.~\ref{sec:intro}, when dealing with spin glasses we need to enlarge our analysis to the $s$-replicated version of the configuration space. To this aim, we use the following
\begin{Def}
	Let $\Sigma_N^{(s)}=( \Sigma_N)^{\otimes s}\equiv \{-1,+1\}^{sN}$ be the $s$-replicated configuration space. We denote with $\bb {\underline{\sigma}}= (\bb\sigma^1,\dots ,\bb\sigma^s)\in \Sigma^{(s)}_N$ the global configuration of the replicated system.
	The space $\Sigma_N^{(s)}$ is naturally endowed with the $s$-replicated Boltzmann-Gibbs measure associated to the partition function
	\begin{equation}
		Z^{(s)}_{N,p, \bb J } (t,x)= \sum_{\bb {\underline{\sigma}} \in \Sigma_N^{(s)}} \exp\Big(\sqrt{\frac{ t p!}{2 N^{p-1}}}\sum_{a=1}^s\sum_{1\le i_1<\dots<i_p \le N}^N J_{i_1 \dots i_p}\sigma_{i_1}^{(a)}\dots \sigma_{i_p}^{(a)} +\sqrt x \sum_{a=1}^s \sum_{i=1}^N J_i \sigma_i ^{(a)}\Big).
	\end{equation}
	We will denote with $B^{(s)}_{N,p, \bb J } (t,x)$ the Boltzmann factor appearing in the $s$-replicated partition function. Given an observable $O:\Sigma_N^{(s)}\to \mR$ on the replicated space, the Boltzmann average w.r.t. the $s$-replicated partition function is
	\begin{equation}
				\Omega^{(s)}_{t,x}(O)= \frac{\sum_{\bb {\underline{\sigma}}} O(\bb {\underline{\sigma}} ) B_{N,p, \bb J } ^{(s)}(t,x)}{ \sum_{\bb {\underline{\sigma}}}  B_{N,p, \bb J }^{(s)} (t,x)}.
	\end{equation}
\end{Def}

\begin{Rmk}
	Clearly, the thermodynamics of the original model is recovered with $t=\beta^2$ and $x=0$.
\end{Rmk}

\begin{Rmk}
	Since replicas are independent $Z^{(s)}_{N,p, \bb J } (t,x)\equiv (Z_{N,p, \bb J } (t,x))^s$.
\end{Rmk}

In the following, in order to lighten the notation, the replica index $s$ of the Boltzmann average $\Omega^{(s)}_{t,x}$ can be dropped, since it is understood directly from the function to be averaged.

\begin{Def}
	Given an observable $O:\Sigma_N^{(s)}\to \mR$ on the replicated space, the quenched average is defined as
	\begin{equation}
		\langle O \rangle _{t,x} = \mathbb E_{\bb J} \Omega_{t,x}(O).
	\end{equation}
\end{Def}

\begin{Rmk}
	In the last definition, the average $\mathbb E_{\bb J}$ is again the expectation value performed over all the quenched disorder, thus including the auxiliary random variables in the interpolating setup.
\end{Rmk}

\begin{Def}\label{Def:Overlap_pSK}
	The order parameter for the $p$-spin SK model is the replica overlap
	\begin{equation}
		q_{ab}=\frac 1N \sum_{i=1}^N \sigma_i ^{(a)} \sigma_i ^{(b)},
	\end{equation}
	where $\bb \sigma^{(a)}$ and $\bb \sigma^{(b)}$ are two generic configurations of different replicas of the system.
\end{Def}

We can now focus on the PDE approach to the statistical mechanics of the $p$-spin SK model. To this aim, we compute the spacetime derivative of the quenched intensive pressure, as given in the following

\begin{Def}\label{Def:GuerraAction_pSK}
	For all $p \ge 2$, Guerra's action functional is defined as
	\begin{equation}
		\label{eq:GuerraAction_pSK}
		S_{N,p}(t,x)= 2A_{N,p}(t,x)-x-\frac t2.
	\end{equation}
\end{Def}

\begin{Lem}\label{Prp:Sderivs_pSK}
	The spacetime derivatives of the Guerra's action functional read as
	\begin{eqnarray}
		\partial_ t S_{N,p}(t,x)&=&-\frac 12 \langle q_{12}^p\rangle_{t,x} + R_N(t,x), \label{eq:SNt_derrida}
		\\
		\partial_ x S_{N,p}(t,x)&=&- \langle q_{12}\rangle_{t,x}. \label{eq:SNx_derrida}
	\end{eqnarray}
	where $R_N(t,x)$ takes into account the contributions coming from \eqref{eq:subleading} and vanishing in the $N\to\infty $ limit. 
\end{Lem}
The proof of this lemma can be found in Appendix \ref{AA:Sderivs_pSK}.

\begin{Lem}\label{Prp:Streaming_pSK}
	Given an observable $O:\Sigma^{(s)}_N \to \mR $ on the replicated space, the following streaming equation holds:
	\begin{equation}
		\begin{split}
			\partial_x \langle O(\bb{\underline\sigma})\rangle_{t,x}&= \frac  N2 \sum_{a,b=1}^s \langle O(\underline{\bb \sigma})q_{ab}\rangle_{t,x}-s N \sum_{a=1}^s \langle O(\underline{\bb\sigma})q_{a,s+1}\rangle_{t,x}\\
			&-\frac s2 N \langle O(\bb{\underline{\sigma}})\rangle_{t,x}+\frac{s(s+1)}{2} \langle O(\bb{\underline{\sigma}})q_{s+1,s+2}\rangle_{t,x}.
		\end{split}
	\end{equation}
\end{Lem}

\begin{proof}
	The proof is long and rather cumbersome, so we will just give a sketch. First of all, we recall that
	\begin{equation}
		\langle O(\bbu \sigma)\rangle_{t,x}=\mathbb E_{\bb J} \frac1{Z_{N,p, \bb J }(t,x)^s} \sum_{\bb {\underline{\sigma}} \in  \Sigma_N^{(s)}}  O(\bbu \sigma)B_{N,p, \bb J } ^{(s)}(t,x).
	\end{equation}
	When taking the $x$-derivative of this quantity, we will get two contributions: the first one follows from the derivative of $B_{N,p, \bb J }^{(s)}(t,x)$, and the second one follows from the derivative of $1/Z_{N,p, \bb J }^s$ (which results in adding a new replica). In quantitative terms:
	\begin{equation}
		\begin{split}			
		\partial_x \langle O(\bbu \sigma)\rangle_{t,x} = \frac1{2\sqrt x}&\mathbb E_{\bb J}\Bigl(\sum_{a=1}^s \sum_{i=1}^N J_i \frac1{Z_{N,p, \bb J }(t,x)^s}\sum_{\bb \sigma^{(1)}}\dots \sum_{\bb \sigma^{(s)}}O(\bbu\sigma) \sigma_i ^{(a)} B_{N,p, \bb J } ^{(s)}(t,x)\\
		&-s \sum_{i=1}^N J_i \frac1{Z_{N,p, \bb J }(t,x)^{s+1}}\sum_{\bb \sigma^{(1)}}\dots \sum_{\bb \sigma^{(s+1)}}O(\bbu\sigma)\sigma_i^{(s+1)}B_{N,p, \bb J }^{(s+1)}(t,x)\Bigr).
		\end{split}
	\end{equation}
	The presence of $J_i$ in both terms of the right-hand-side can be carried out by applying the Wick-Isserlis theorem. Each $J_i$-derivative would result in two different contributions, and its action on the denominators (involving the partition functions) will again result in the appearance of Boltzmann averages with more replicas. Further, the explicit $x$-dependence of the derivative precisely cancels (since the $J_i$-derivative will produce factors proportional to $\sqrt x$). Indeed, after all the computations and recalling $\langle \cdot \rangle_{t,x}=\mathbb E_{\bb J}\Omega_{t,x}(\cdot)$, we get
	\begin{equation}
		\begin{split}
			\partial_x \langle O(\bbu \sigma)\rangle_{t,x} &=\frac12 \sum_{a,b=1}^s\sum_{i=1}^N  \langle O\bbu\sigma)\sigma_i ^{(a)}\sigma_i ^{(b)}\rangle_{t,x}-\frac s2 \sum_{a=1}^s\sum_{i=1}^N\langle O(\bbu\sigma)\sigma_i ^{(a)} \sigma_i ^{(s+1)}\rangle_{t,x}\\
			&-\frac s2 \sum_{s=1}^{s+1}\sum_{i=1}^N \langle O(\bbu \sigma)\sigma_i^{(a)} \sigma_i ^{(s+1)}\rangle_{t,x}+\frac{s(s+1)}{2}\sum_{i=1}^N \langle O(\bbu \sigma)\sigma_i^{(s+1)}\sigma_i^{(s+2)}\rangle_{t,x}.
		\end{split}
	\end{equation}
	Recalling Def. \ref{Def:Overlap_pSK} and after some rearrangements of the quantities, we get the thesis.
\end{proof}

\begin{Cor}
	For each $l\in \mathbb N$, the following equality holds:
	\begin{equation}
		\label{eq:StreamingQ_pSK}
		\partial_x \langle q_{12}^l\rangle_{t,x}= N\bigl(\langle q_{12}^{l+1} \rangle_{t,x}-4 \langle q^l_{12}q_{23}\rangle_{t,x}+3 \langle q_{12 }^l q_{34} \rangle_{t,x}\bigr).
	\end{equation}
\end{Cor}
\begin{proof}
	The proof works simply by putting $O(\bbu\sigma)= q_{12}^l$ (which is a function of two replicas of the system) in Prop. \ref{Prp:Streaming_pSK}.
\end{proof}

In order to proceed, we have now to make some physical assumptions on the model. As standard in spin-glass theory, the simplest requirement is the RS in the thermodynamic limit. In fact, as we are going to show, this makes the PDE approach feasible, due to the fact that we can express non-trivial expectation values of function of the replicas in a very simple form.

\begin{Prp}\label{Prp:degreep_equality_pSK}
	For the interpolated Derrida model (\ref{def:DM}), the following equality holds:
	\begin{equation}
		\langle q_{12}^p\rangle_{t,x}= \left(\frac1N \partial_x + \langle q_{12}\rangle_{t,x}\right)^{p-1}\langle q_{12}\rangle_{t,x}+Q_N^{(p-1)}(t,x),
	\end{equation}
	where $Q_N^{(p-1)}(t,x)$ vanishes in the $N\to\infty$ limit and under the RS assumption.
\end{Prp}
\begin{proof}
 Let us consider the $x$-derivative of $\langle q_{12}\rangle$ and try to rearrange the first contribution:
\begin{equation}
	\begin{split}
		\langle q_{12}^{l}q_{23}\rangle_{t,x}=\langle q_{12}^{l}\Delta({q_{23}})\rangle_{t,x}+\langle q_{23}\rangle_{t,x}\langle q_{12}^l \rangle_{t,x},
	\end{split}
\end{equation}
where $\Delta(q_{ab})= q_{ab}-\langle q_{ab}\rangle_{t,x}$ $\forall a,b$ is the fluctuation of the overlap w.r.t its thermodynamic value. Further
\begin{equation}
	\begin{split}
		\langle q_{12}^{l}q_{23}\rangle_{t,x}&= \langle \Delta(q_{12}^{l})\Delta({q_{23}})\rangle_{t,x}+\langle q_{12}^l \rangle_{t,x}\langle \Delta(q_{23}) \rangle_{t,x}+\langle q_{23}\rangle_{t,x}\langle q_{12}^l \rangle_{t,x}\\&=\langle q_{12}\rangle_{t,x}\langle q_{12}^l \rangle_{t,x}+R^{(1,l)}_{N}(t,x),
	\end{split}
\end{equation}
where, $R^{(1,l)}_N(t,x)$ represents the terms involving the fluctuation functions of the overlap. In the last equality, we also used the fact that $\langle q_{23}\rangle_{t,x}=\langle q_{12}\rangle_{t,x}$ since the average is independent on the replicas labelling. 
The last term in \eqref{eq:StreamingQ_pSK} has a similar expansion:
\begin{equation}
	\begin{split}
		\langle q_{12}^{l}q_{34}\rangle_{t,x}=\langle q_{12}\rangle_{t,x}\langle q_{12}^l \rangle_{t,x}+R^{(2,l)}_{N}(t,x),
	\end{split}
\end{equation}
thus we finally get
\begin{equation}
	\begin{split}
		\langle q_{12}^{l+1}\rangle_{t,x}=\frac1N\partial_x \langle q_{12}^l \rangle_{t,x}+\langle q_{12}\rangle_{t,x}\langle q_{12}^l \rangle_{t,x}+R_{N}^{(l)}(t,x),
	\end{split}
\end{equation}
where $R_{N}^{(l)}(t,x)=4R^{(1,l)}_{N}(t,x)-3R^{(2,l)}_{N}(t,x)$. We can then express higher moments of the overlap in terms of lower ones: 
\begin{equation}
	\begin{split}
		\langle q_{12}^{l+1}\rangle_{t,x}=
		\left(\frac1N\partial_x+ \langle q_{12}\rangle_{t,x}\right) \langle q_{12}^l \rangle_{t,x} +R^{(l)}_{N}(t,x).
	\end{split}
\end{equation}
Iterating this procedure from $l=1$ up to $l=p-1$, we obtain 
\begin{equation}
	\begin{split}
		\label{eq:qp_derrida}
		\langle q_{12}^{p}\rangle_{t,x}=
		\left(\frac1N\partial_x+ \langle q_{12}\rangle_{t,x}\right)^{p-1} \langle q_{12} \rangle_{t,x} +Q^{(p-1)}_{N}(t,x),
	\end{split}
\end{equation}
where $Q^{(p-1)}_{N}(t,x)$ collects all the terms involving $R_N^{(l)}(t,x)$, and thus vanishes in the $N\to\infty$ limit.
\end{proof}
\par\medskip
At this point, we have at our disposal all the ingredients needed for making explicit our approach. 
Using \eqref{eq:qp_derrida} in \eqref{eq:SNt_derrida}, we get
\begin{equation}\label{lim_DtS_1}
	\partial_t S_{N,p}(t,x)=-\frac12\left(\frac1N \partial_x + \langle q_{12}\rangle_{t,x}\right)^{p-1} \langle q_{12}\rangle_{t,x}+R_N(t,x)-\frac12 Q_N^{(p-1)}(t,x).
\end{equation}
Deriving \eqref{lim_DtS_1} with respect to the spatial coordinate $x$, we have
\begin{equation*}
	\partial _t \partial_x S_{N,p} (t,x) =-\frac12\partial_x\left(\frac1N \partial_x + \langle q_{12}\rangle_{t,x}\right)^{p-1} \langle q_{12}\rangle_{t,x}+V_N(t,x),
\end{equation*}
where $V_N(t,x):=-\partial_x(R_N(t,x)-\frac12 Q_N^{(p-1)}(t,x))$, vanishing in the $N\to\infty$ limit. We can then write the following equation 
\begin{equation}\label{BurgerN}
	\partial _t \langle q_{12}\rangle_{t,x}-\frac12\partial_x\left(\frac1N \partial_x + \langle q_{12}\rangle_{t,x}\right)^{p-1} \langle q_{12}\rangle_{t,x}=V_N(t,x).
\end{equation}
On the l.h.s. we recognize a Burgers hierarchy structure, while on the r.h.s. we have a source term (which further vanishes in the thermodynamic limit). The equilibrium dynamics of the Derrida models is then realized taking the limit $N\to\infty$, as summarized in the following
\begin{Thm}
The expectation value of the order parameter for the $p$-spin Derrida model under the RS ansatz is given by the function $\bar q(\beta)= u(\beta^2,0)$, where $u(t,x)$ is the solution of the inviscid limit of the $p$-th element Burgers hierarchy with initial profile \eqref{eq:initialvalue_pSK}, {\it i.e.}
	\begin{equation}\label{eq:BurgersPDC_pSK}
		\begin{cases}
			\partial _t u(t,x) -\frac 12 \partial_x u^p(t,x) =0\\
			u(0,x)=\mathbb E_J \tanh^2 (\sqrt x J)
		\end{cases}.
	\end{equation}
\end{Thm}
\begin{proof}
By making the limit of the previous equation \eqref{BurgerN} for $N\to\infty$ and recalling that $V_N(t,x)\rightarrow 0$ for $N\to\infty$ we get 
\begin{equation}\label{pde}
\partial_t u(t,x)-\frac 1 2 \partial_x u^p(t,x)=0,
\end{equation}
where $u(t,x):=\lim_{N\to \infty}\langle q_{12}\rangle_{t,x}$. The initial profile of the Cauchy problem associated to the PDE \eqref{pde} is easily determined, since for $t=0$ the partition function reduces to a $1$-body problem. Thus, we have to compute $u(0,x)=\lim_{N\to\infty} \langle q_{12}\rangle_{0,x}$. To this aim, we start from the partition function evaluated at $t=0$, which is
\begin{equation}
	\begin{split}
		Z_{N,p,\bb J}(0,x)&=\sum_{\bb \sigma}\exp\bigl(\sqrt x \sum_{i=1}^N J_i\sigma_i\bigr)= \prod_{i=1}^N 2\cosh (\sqrt x J_i).
	\end{split}
\end{equation}
Taking the logarithm and averaging over the quenched disorder $\bb J$, we have the intensive pressure:
\begin{equation}
	\begin{split}
		A_{N,p}(0,x)&=\frac1N \qavJ \log \prod_{i=1}^N 2\cosh (\sqrt x J_i)= \log 2+\frac2N \sum_{i=1}^N\qavJ \log \cosh (\sqrt x J_i).
	\end{split}
\end{equation}
Recalling that the Guerra's action is defined as \eqref{eq:GuerraAction_pSK} and that the $J_i$ are i.i.d. so the sum of quenched averages of functions of $J_i$ is $N$ times the average w.r.t. a {\it single} quenched variables $J\sim \mathcal N (0,1)$, we get 
\begin{equation}
	S_N (0,x)=2\log 2+2\mathbb E_J \log\cosh (\sqrt x J)-x.
\end{equation}
Finally, taking the derivative w.r.t. the spatial coordinate, we finally have the initial profile for the overlap expectation value, which reads
\begin{equation}\label{eq:initialvalue_pSK}
	u(0,x)=\lim_{N\to\infty}\langle q_{12}\rangle_{0,x}= \mathbb E_J \tanh^2 (\sqrt x J).
\end{equation}
Here, we used again the Wick-Isserlis theorem for normally distributed random variables. Putting together \eqref{pde} and \eqref{eq:BurgersPDC_pSK}, we get the thesis.
\end{proof}

\begin{Prp}
	The implicit solution of the inviscid Burgers hierarchy \eqref{eq:BurgersPDC_pSK} is the self-consistency equation for the order parameter $\bar q(\beta)$ for the $p$-spin model under the RS ansatz.
\end{Prp}
\begin{proof}
	Let us rewrite the differential equation \eqref{eq:BurgersPDC_pSK} as
	\begin{equation}
		\partial _t u -\frac p2  u^{p-1}\partial_x u =0.
	\end{equation}
	This is a non-linear wave equation and, as well-known, it admits a solution of the form $u(t,x)= u_0 (x-v(t,x)t)$, where $u_0$ is the initial profile and $v(t,x)$ is the effective velocity. For the case under consideration, we have $v(t,x)=-\frac p2 u^{p-1}(t,x)$, thus
	\begin{equation}
		u(t,x)=\mathbb E_J \tanh^2 \left(\sqrt{x+ t \frac p2 u(t,x)^{p-1}} J\right).
	\end{equation}
	Recalling that $\bar q(\beta)=u(\beta^2,0)$, we finally have
	\begin{equation}\label{eq:pspin_sc}
		\bar q = \mathbb E_J \tanh^2 \left(\beta \sqrt{\tfrac p2 \bar q^{p-1}} J\right),
	\end{equation}
	which is precisely the self-consistency equation for the $p$-spin glass model, as reported also in \cite{agliari2012notes}.
\end{proof}

\begin{Cor}
	The (ergodicity breaking) phase transition of the $p$-spin model coincides with the gradient catastrophe of the Cauchy problem \eqref{eq:BurgersPDC_pSK}, and the critical temperature is determined by the system
	\begin{equation}\label{eq:psk_criticality}
		\begin{cases}
			T_c= \sqrt{ - F'(\bar \xi)}\\
			\bar \xi= \frac{F(\bar\xi)}{F'(\bar \xi)},
		\end{cases}
	\end{equation}
	where $F(\xi)= -\tfrac p2\mathbb E_J \tanh^2 (\sqrt x J)$.
\end{Cor}

\begin{proof}
	The determination of the critical temperature can be achieved with the usual analysis of intersecting characteristics of the Cauchy problem \eqref{eq:BurgersPDC_pSK}, and follows the same lines of \cite{fachechi2021pde}.
\end{proof}

As a comparison, in Fig. \ref{fig:Derrida_solutions}, we reported the solutions of the self-consistency equations \eqref{eq:pspin_sc} for $p=2,\dots,8$ (solid curves), and the critical temperatures as predicted by the system \eqref{eq:psk_criticality} (dashed lines).

\begin{figure}
	\centering
	\includegraphics[width=0.5\textwidth]{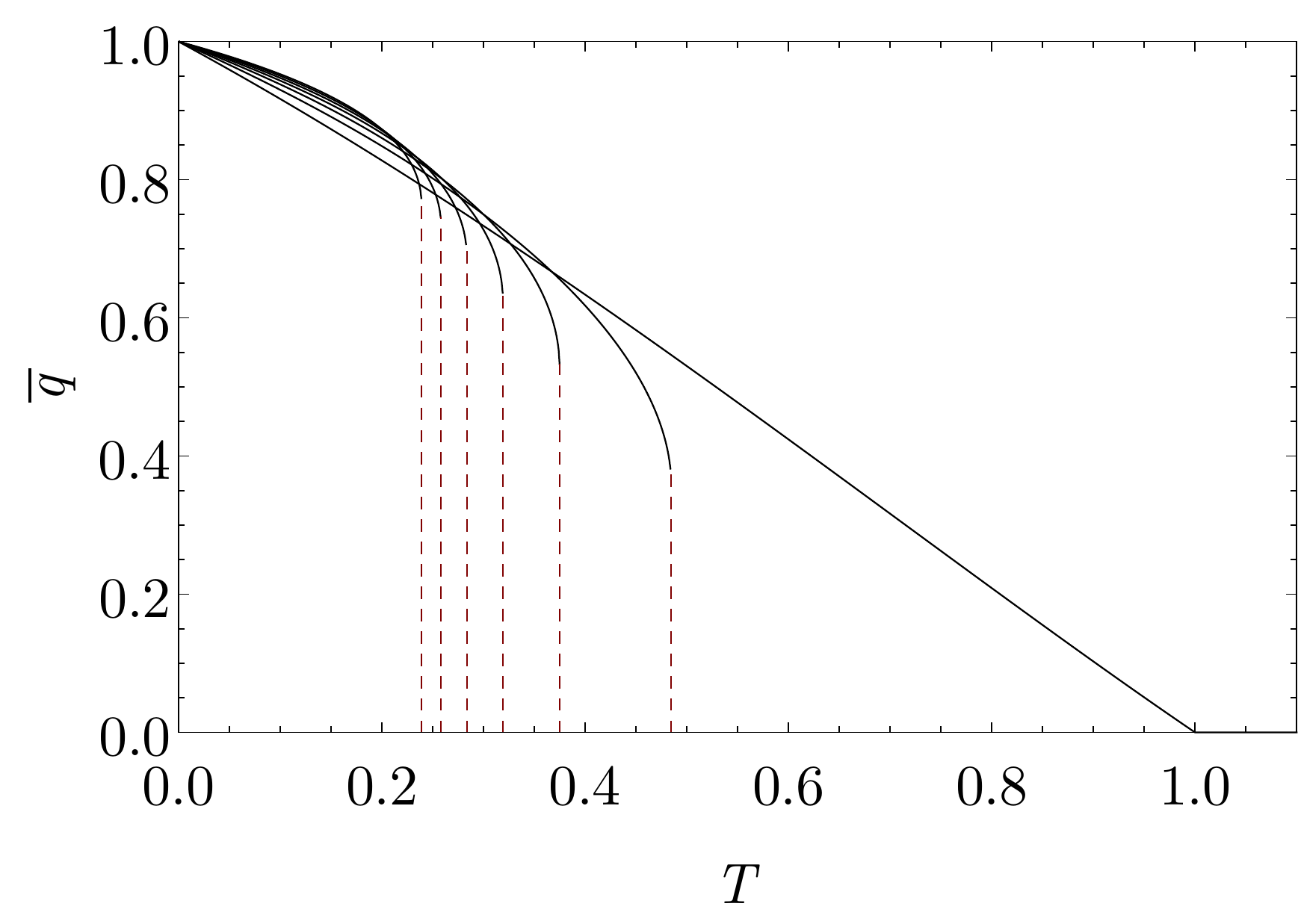}
	\caption{Solutions of the self-consistency equations \eqref{eq:pspin_sc} for $p = 2,\dots,8$ (solid curves), and the critical temperatures as predicted by the system \eqref{eq:psk_criticality} (dashed lines).}\label{fig:Derrida_solutions}
\end{figure}

\section{Application to Dense Associative Memories}
Going beyond the pure spin-glass case, in this Section we will approach the DAMs that are the main focus of this work.

\begin{Def}\label{H_DHN}
Let $\bb\sigma $ be the generic point in the configuration space $\Sigma_N \equiv \{-1,+1\}^N$ of the system. Given $K$ random patterns $\{\bxi^\mu\}_{\mu=1}^K$ each made of $N$ i.i.d. binary entries drawn with equal probability $P(\xi^\mu_i=-1)=P(\xi^\mu_i=1)=\frac 1 2$ $\forall i=1,\dots,N$, the Hamiltonian of the $p$-th order DAM is 
\begin{equation}\label{Hp}
H_{N,p, \bxi,K}(\bsigma):=-\frac{1}{N^{p-1}}\sum_{\mu=1}^K\sum_{i_1,\dots,i_p}^N\xi^\mu_{i_1}\dots\xi^\mu_{i_p}\sigma_{i_1}\dots\sigma_{i_p}.
\end{equation}
\end{Def}
\begin{Rmk}
The normalization factor $\frac{1}{N^{p-1}}$ ensures the linear extensivity of the Hamiltonian, in the volume of the network $N$, i.e. $\lim_{N\to\infty}\left\vert\frac{H_{N,p, \bxi, K}}{N}\right\vert\in(0,+\infty)$.
\end{Rmk}

As already stated in Sec.~\ref{sec:intro}, DAMs with $p$-order interactions are able to store at most a number of patterns $K=\alpha N^{p-1}$, with $\alpha<\alpha_c$ \cite{baldi1987number,bovier2001spin} and, in the following, we will study the model in two different regimes, that is, setting $K$ finite and setting $K$ such that $\alpha$ is finite, corresponding to, respectively, a simple and a complex scenario.

\subsection{Low storage}\label{LS_Hop}
Let us start the analysis  of the network in a low-load regime, storing a finite number of patterns. Again, the goal is to use interpolation techniques and derive PDEs equations able to describe the thermodynamics of the system. To do this, let start by defining 
\begin{Def}\label{m}
The order parameters used to describe the macroscopic behavior of the model are the so-called Mattis magnetizations, defined as
\begin{equation}
m_\mu(\bsigma):=\frac 1 N\sum_{i=1}^N \xi^\mu_i\sigma_i \quad \forall \mu=1,\dots,K
\end{equation} 
measuring the overlap between the network configuration and the stored patterns. 
\end{Def}
\begin{Rmk}
	The Hamilton function \eqref{Hp} in terms of the Mattis magnetizations is 
	$$
	H_{N,p,\bb\xi,K}(\bb\sigma)=-N \sum_{\mu=1}^K\big(m_\mu(\bb \sigma)\big)^p.
	$$
\end{Rmk}
Next, we define the basic objects of our investigations within the interpolating framework.

\begin{Def} \label{def:DAM_p}
	Given $(t,\vx)\in \mR^{K+1}$, the spacetime Guerra's interpolating partition function for the DAM model (in the low-load regime) reads as
	\begin{eqnarray}
		Z_{N,p,\bb\xi,K}(t,\vx)&=&\sum_{\bb\sigma\in\Sigma_N}\exp\big(-H_{N,p,\bb\xi,K}(t,\vx)\big),\label{eq:dense_guerra_1} \\
		H_{N,p,\bb\xi,K}(t,\vx)&=& t N \sum_{\mu =1}^Km_\mu (\bsigma)^p-N\sum_{\mu =1}^K x_\mu m_\mu (\bsigma).
	\end{eqnarray}
\end{Def}
\begin{Rmk}
	Clearly, the spacetime Guerra's interpolating partition function recovers the one related to the DAM by setting $t=-\beta$ and $\vx=\bb 0$.
\end{Rmk}

\begin{Rmk}
We recall that, for the CW model, the Guerra mechanical analogy consists in interpreting the statistical pressure as the Burgers hierarchy describing the motion of viscid non-linear waves in $1 + 1$-dimensional space. In the case of the DAMs, we have $K$ Mattis magnetizations, and the dual mechanical system describes non-linear waves travelling in a $K + 1$-dimensional space. 
\end{Rmk}

\begin{Def}
For each configuration $\bsigma \in \Sigma_N$ of the system, the Boltzmann factor corresponding to the partition function \eqref{eq:dense_guerra_1} is 
\begin{equation}\label{eq:dense_guerra_5}
	B_{N,p,\bb\xi,K}( t,\vx)=\exp\Big(-t N \sum_{\mu =1}^Km_\mu (\bsigma)^p+N\sum_{\mu =1}^K x_\mu m_\mu (\bsigma)\Big).
\end{equation}
\end{Def}
\begin{Prp}
	The first-order spacetime derivatives of the Guerra intensive pressure associated to the partition function \eqref{eq:dense_guerra_1} read as
	\begin{eqnarray}
	\partial_t A_{N,p,\bb\xi,K}(t,\vx)&=&-\sum _{\mu=1}^K\average{t,\vx}{m_\mu (\bsigma)^p},\label{eq:dense_guerra_3}\\
	\partial_{\mu} A_{N,p,\bb\xi,K}(t,\vx)&=&\average{t,\vx}{m_\mu (\bsigma)}.\label{eq:dense_guerra_4}
	\end{eqnarray}
	where $\partial_\mu:=\partial_{x^\mu}$.
\end{Prp}
\begin{proof}
	Recalling the definition of the intensive pressure $ A_{N,p,\bb\xi,K} (t,\vx)= \frac1N \log \sum_{\bb\sigma\in\Sigma_N} 	B_{N,p,\bb\xi,K}( t,\vx)$, along with eq.~\ref{eq:dense_guerra_5}, the proof follows
	straightforward computations. The temporal derivative reads 
	\begin{equation*}
	\begin{split}
	\partial _t A_{N,p,\bb\xi,K}(t,\vx)&=\frac1N  Z^{-1}_{N,p,\bb\xi,K}(t,\vx)\sum_{\bb\sigma\in\Sigma_N}\Big(-N\sum_{\mu=1}^Km_\mu (\bsigma)^p\Big)B_{N,p,\bb\xi,K}( t,\vx)=-\sum _{\mu=1}^K\average{t,\vx}{m_\mu (\bb\sigma)^p}.
	\end{split}
	\end{equation*}
	while the spatial derivative reads
		\begin{equation*}
	\begin{split}
	\partial _{\mu} A_{N,p,\bb\xi,K}(t,\vx)&=\frac1N  Z^{-1}_{N,p,\bb\xi,K}(t,\vx)\sumsigma\Big(N m_\mu (\bsigma)\Big)B_{N,p,\bb\xi,K}( t,\vx)=\average{t,\vx}{m_\mu (\bb\sigma)}.
	\end{split}
	\end{equation*}
\end{proof}

\begin{Prp}
The higher (non-centered) momenta of the Mattis magnetizations are realized as
\begin{equation}
\label{eq:dense_guerra_7}
\average{t,\vx}{m_\mu^{s+1}}=\Big(\frac1N \partial_\mu+\average{t,\vx}{m_\mu}\Big)\average{t,\vx}{m_\mu ^s},
\end{equation}
for each integer $s\ge1$.
\end{Prp}
\begin{proof}
We start by computing the spatial derivative of the Mattis magnetizations expectation value:
\begin{equation*}
\begin{split}
\partial_\nu \average{t,\vx}{m_\mu^s}&=\partial_\nu\Big( Z^{-1}_{N,p,\bb\xi,K} (t,\vx)\sumsigma m_\mu(\bsigma)^s B_{N,p,\bb\xi,K}(t,\vx) \Big)=\\&=N Z^{-1}_{N,p,\bb\xi,K} (t,\vx)\sumsigma m_\mu(\bsigma)^s m_\nu (\bsigma) B_{N,p,\bb\xi,K}(t,\vx)\\&-NZ^{-1}_{N,p,\bb\xi,K} (t,\vx)\sumsigma m_\mu(\bsigma)^s B_{N,p,\bb\xi,K}(t,\vx)\cdot Z^{-1}_{N,p,\bb\xi,K} (t,\vx)\sumsigmap m_\nu (\bsigma') B_{N,p,\bb\xi,K}(t,\vx)=\\&=
N\average{t,\vx}{m_\mu^s m_\nu}-N\average{t,\vx}{m_\mu^s}\average{t,\vx}{m_\nu}.
\end{split}
\end{equation*}
In particular, for $\nu=\mu$, we have
$$
\partial_\mu \average{t,\vx}{m_\mu^s}=N[\average{t,\vx}{m_\mu^{s+1}}-\average{t,\vx}{m_\mu^s}\average{t,\vx}{m_\mu}].
$$
Expressing the higher order moment in terms of the other quantities, we reach the thesis.
\end{proof}
By calling $u^{(s)}_\mu(t,\vx):=\average{t,\vx}{m_\mu(\bsigma)^s}$, we can express all $u^{(s)}_\mu(t,\vx)$ in terms of $u_\mu^{(1)}(t,\vx):=u_\mu(t,\vx)$ for each $s>1$. Indeed
\begin{equation}
\label{eq:dense_guerra_8}
u_\mu ^{(s+1)}(t,\vx)=\Big(\frac1N \partial_\mu+u_\mu(t,\vx)\Big)u_\mu ^{(s)}(t,\vx)=\Big(\frac1N \partial_\mu+u_\mu(t,\vx)\Big)^s u_\mu(t,\vx).
\end{equation}
To simplify the notation, we define the operator $D_\mu:= \frac1N \partial_\mu+u_\mu(t,\vx)$.

\begin{Thm}
The expectation value of the Mattis magnetizations of the interpolated DAM model (\ref{def:DAM_p}) satisfies the non-linear evolutive equations
\begin{equation}
\label{eq:dense_guerra_9}
\partial_t u_\mu (t,\vx)=-\sum_{\nu=1}^K \partial_\mu D_\nu ^{p-1}u_\nu(t,\vx).
\end{equation}
\end{Thm}
\begin{proof}
First, we put $s=p-1$ in \eqref{eq:dense_guerra_8}, so that
$$
u_\nu ^{(p)}(t,\vx)=D_\nu^{p-1} u_\nu(t,\vx).
$$
Now, recall that $u_\nu ^{(p)}(t,\vx)=\average{t,\vx}{m_\nu(\bsigma)^p}$ and $\partial_t A_{N,p,\bb\xi,K}(t,\vx)=-\sum _{\mu=1}^K\average{t,\vx}{m_\mu (\bsigma)^p}$, thus
$$
\partial_t A_{N,p,\bb\xi,K}(t,\vx)=-\sum_{\nu=1}^K u_\nu ^{(p)}(t,\vx)=-\sum_{\nu=1}^KD_\nu^{p-1} u_\nu(t,\vx).
$$
Taking the derivative $\partial_\mu$, commuting $\partial_t$ and $\partial _\mu$ and recalling that $\partial_\mu A_{N,p,\bb\xi,K}(t,\vx)=\average{t,\vx}{m_\mu (\bsigma)}=u_\mu(t,\vx)$, we directly reach the thesis.
\end{proof}

\begin{Lem}
The evolutive equations \eqref{eq:dense_guerra_9} can be linearized by means of the Cole-Hopf transform.
\end{Lem}

\begin{proof}
To proof our assertion, we use the basic identities
\begin{eqnarray*}
\Big(\partial_\mu +\frac{\partial_\mu \Psi}{\Psi}\Big)^s\frac{\partial_\mu \Psi}{\Psi}&=& \frac{\partial_\mu ^{s+1}\Psi}{\Psi},\\
\partial _t \frac{\partial_\mu \Psi}{\Psi}&=&\partial_\mu\frac{\partial_t \Psi}{\Psi}.
\end{eqnarray*}
Performing the multi-dimensional Cole-Hopf transform $u_\mu(t,\vx)=\frac1N \partial_\mu (\log \Psi)$, we have
\begin{equation*}
	\begin{split}
	\frac1N\partial_t \frac{\partial_\mu \Psi}{\Psi}=-\frac1{N^p}\sum_{\nu=1}^K\partial_\mu \Big(\partial_\mu +\frac{\partial_\mu \Psi}{\Psi}\Big)^{p-1}\frac{\partial_\mu \Psi}{\Psi},
	\end{split}
\end{equation*}
and using the previous properties we have
\begin{equation*}
\partial_\mu\Big( \frac{\partial_t \Psi}{\Psi}+\sum_{\nu=1}^K \frac{\partial_\mu^p \Psi}{\Psi}\Big)=0,
\end{equation*}
Setting the argument of the spatial derivative to zero and assuming $\Psi\neq0$, we have
$$
\partial_t \Psi+\frac1{N^{p-1}} \sum_{\nu=1}^K \partial_\nu^p \Psi=0.
$$
\end{proof}

\begin{Rmk}
	In the proof, the function $\Psi$ is nothing but Guerra's interpolating partition function, as can be understood by comparing the definitions $u_\mu(t,\vx)=\frac1N \partial_\mu\log \Psi(t,\vx)$ and $u_\mu(t,\vx) = \partial_\mu A_{N,p,\bb\xi,K}(t,\vx)$. Indeed, by computing the derivatives of the partition function we easily get
\begin{eqnarray*}
	\partial_t Z_{N,p,\bb\xi,K} (t,\vx)&=& -N\sum_{\mu=1}^K  \sumsigma  m_\mu^p (\bsigma) B_{N,p,\bb\xi,K}(t,\vx),\\
	\partial_\mu^p Z_{N,p,\bb\xi,K} (t,\vx)&=& N^p \sumsigma m_\mu^p (\bsigma) B_{N,p,\bb\xi,K}(t,\vx).
\end{eqnarray*}
A direct comparison shows that Guerra's interpolating partition function satisfies the same differential equation of the $\Psi$ potential.
\end{Rmk}

\begin{Rmk}
The case $K=1$ corresponds to the $p$-spin CW model treated in \cite{fachechi2021pde}. Indeed, the partition function of the system can be handled as
\begin{equation*}
	\begin{split}
	Z_{N,p,\bb\xi,K=1} (t,\vx)&=\sumsigma\exp\Big(-t N (\tfrac1N \sum_i \xi^1_i\sigma_i)^p+N x (\tfrac1N \sum_i \xi^1_i\sigma_i) \Big)=\\
	&=\sumsigma\exp\Big(-t N (\tfrac1N \sum_i\sigma_i)^p+N x (\tfrac1N \sum_i \sigma_i) \Big)=Z_{N,p} ^{(CW)}(t,\vx),
	\end{split}
\end{equation*}
where we used the invariance of the partition function under the transformation $\sigma_i \to \xi^1_i \sigma_i$. 
In this particular case we recover the Burgers hierarchy with viscosity parameter $1/N$: calling $x_1=x$ and $u_1(t,\vx)=u(t,\vx)$, the family \eqref{eq:dense_guerra_9} reduces to
$$
\partial_t u+\partial_x \Big(\frac1N\partial_x+u\Big)^{p-1} u=0.
$$
\end{Rmk}
Within this framework, we generate multi-dimensional generalization of Burgers hierarchy, see Appendix \ref{BB:particular cases} for further details and examples.

\subsection{High-storage}
Here we will study the $p$-spin DAMs in the high-load regime $\lim_{N\to \infty}\frac{K}{N^{p-1}}=\alpha>0$ for even $p$, which now behaves as a complex system with global non-trivial properties. Let us start by observing that, in this case, the partition function related to the Hamiltonian \eqref{Hp} can also be written in the following form\footnote{Notice the little abuse of notation in the expression $Z_{N,p,\bb\xi,\alpha}(\beta)$: in the subscript, $\alpha$ ismeant as the ratio $K/N^{p-1}$ by-passing the thermodynamic limit}:
\begin{equation}\label{Z_DHN}
\begin{split}
Z_{N,p,\bb\xi,\alpha}(\beta)&:=\sum_{\bsigma\in\Sigma_N}\exp{\left(-\beta H_{N,p,\bb\xi,\alpha}(\bsigma)\right)}=\\&=\sum_{\bsigma\in\Sigma_N}\exp{\Bigg[\frac{\beta}{p!N^{p-1}}\sum_{\mu=1}^K\sum_{i_1,\dots,i_p}^N\xi^\mu_{i_1}\dots\xi^\mu_{i_p}\sigma_{i_1}\dots\sigma_{i_p}
\Bigg]}=\\&=\sum_{\bsigma\in\Sigma_N}\exp{\Bigg[\frac{\beta}{p!N^{p-1}}\sum_{\mu=1}^K\Big(\sum_{i_1,\dots,i_{p/2}}^N\xi^\mu_{i_1}\dots\xi^\mu_{i_{p/2}}\sigma_{i_1}\dots\sigma_{i_{p/2}}\Big)^2\Bigg]},
\end{split}
\end{equation}
and, by Hubbard-Stratonovich transforming, we get
\begin{equation}\label{Z_DHN_g}
\begin{split}
Z_{N,p,\bb\xi,\alpha}(\beta)=\sum_{\bsigma\in\Sigma_N} \prod_{\mu=1}^K \int d\tau_\mu\frac{e^{-\frac{ \tau_\mu^2}{2}}}{\sqrt{2\pi}}\exp{\Bigg(\sqrt{\frac{2\beta}{p!N^{p-1}}}\sum_{i_1,\dots,i_{p/2}}^N\xi^\mu_{i_1}\dots \xi^\mu_{i_{p/2}}\sigma_{i_1}\dots\sigma_{i_{p/2}}\tau_\mu
\Bigg)}.
\end{split}
\end{equation}
Here, we used the index $\alpha$ rather than $K$ in order to distinguish between the partition function of low and high storage regimes.
\begin{Rmk}
	As standard in statistical mechanics of (complex) neural networks, we will assume that a single pattern is candidate to be retrieved, say $\bb\xi^1$. Under this assumption, we can treat separately the Mattis magnetization $m:=m_1$ corresponding to the recalled pattern from those associated to non-retrieved ones.
\end{Rmk}

\begin{Def}
	Given $(t,\vx)\in \mR^{4}$, the spacetime Guerra's interpolating partition function for the DAM model (in the high storage regime) reads as
	\begin{eqnarray}
		Z_{N,p,\bb\xi,\alpha}(t,\vx)&=&\sum_{\bsigma\in\Sigma_N} \prod_{\mu=1}^K \int d\tau_\mu\frac{e^{-\frac{ \tau_\mu^2}{2}}}{\sqrt{2\pi}}\exp\big(-H_{N,p,\bb\xi,\alpha}(t,\vx)\big),\\
		H_{N,p,\bb\xi,\alpha}(t,\vx)&=&-\frac{tN}{2}m^p-\sqrt{\frac{t}{N^{p-1}}}\sum_{\mu=2}^K\sum_{i_1,\dots,i_{p/2}}^N\xi_{i_1}^\mu\dots\xi^\mu_{i_{p/2}}\sigma_{i_1}\dots\sigma_{i_{p/2}}\tau_\mu\\ &-&\sqrt x \sum_{i=1}^N \eta_i \sigma_i -\sqrt{N^{1-p/2}y} \sum_{\mu=1}^K \theta_\mu \tau_\mu -\frac{N^{1-p/2}}{2}(t-t_0+y)\sum_{\mu=1}^K \tau_\mu^2 -\frac N 2  z m,\notag
	\end{eqnarray}
where $\vx =(x,y,z)$.
\end{Def}


\begin{Rmk}
Clearly, the partition function of the original DAM model \eqref{Z_DHN_g} is recovered by setting $t_0=\frac{2\beta}{p!}$ and $(t,x,y,z) =(\frac{2\beta}{p!},0,0,0)$.
\end{Rmk}
In the case under consideration, the Boltzmann average w.r.t. the interpolating measure is
\begin{equation}
\omega_{t,\bb x}(O)=\frac{\sum_{\bsigma}\int d\mu(\btau)\,O (\bb\sigma,\bb\tau) B_{N,p,\bb\xi,\alpha}(t,\vx)}{Z_{N,p,\bb\xi,\alpha}(t,\bb x)},
\end{equation}
where $O(\bb\sigma,\bb\tau)$ is a generic observables in the configuration space of the system, and $B_{N,p}$ is the generalized Boltzmann factor for the $p$-spin model defined as 
\begin{equation*}
\begin{split}
B_{N,p,\bb\xi,\alpha}(t,\vx):&=\exp\biggl(\frac{tN}{2}m^p+\sqrt{\frac{t}{N^{p-1}}}\sum_{\mu=2}^K\sum_{i_1,\dots,i_{p/2}}^N\xi_{i_1}^\mu\dots\xi^\mu_{i_{p/2}}\sigma_{i_1}\dots\sigma_{i_{p/2}}\tau_\mu\\ &+\sqrt x \sum_{i=1}^N \eta_i \sigma_i +\sqrt{N^{1-p/2}y} \sum_{\mu=1}^K \theta_\mu \tau_\mu -\frac{N^{1-p/2}}{2}(t-t_0+y)\sum_{\mu=1}^K \tau_\mu^2 +\frac N 2  z m\biggl).
\end{split}
\end{equation*}
%
As we said in Sec.~\ref{sec:intro}, the high storage regime of associative neural networks exhibits both ferromagnetic and spin-glass features. Thus, besides the usual Mattis magnetisations, we need the overlap for the two sets of relevant variables in the integral formulation of the partition function \eqref{Z_DHN_g}:
\begin{Def}\label{order_para}
The order parameters used to describe the macroscopic behavior of the model are the overlap $m$ (already defined in \eqref{m} and used to quantify the retrieval capability of the network), the replica overlap in the $\bsigma$ variables 
\begin{equation}
	\label{overlap_q}
	q_{12}:=\frac 1 N\sum_{i=1}^N\sigma_i^{(1)}\sigma_i^{(2)},
\end{equation}
and the replica overlap in the $\btau$'s variables
\begin{equation}
	\label{overlap_p}
	p_{12}:=\frac 1 {N^{p/2}}\sum_{\mu=1}^K\tau_\mu^{(1)}\tau_\mu^{(2)}.
\end{equation}
\end{Def}

With all these ingredients at our hand, we now move to the formulation of the PDE duality of DAMs in the high-storage limit.
\begin{Def}
For all even $p\ge 2$, Guerra's action functional is defined as
\begin{equation}\label{S}
S_{N,p,\alpha}(t,\vx):=2A_{N,p,\alpha}(t,\vx)-x.
\end{equation}
\end{Def}

\begin{Lem}\label{derSN}
The partial derivatives of Guerra's action $S_{N,p}(t,\vx)$ can be expressed in terms of the generalized expectations of the order parameters as
\begin{equation}\label{dA_DHN}
\begin{split}
\partial_t S_{N,p,\alpha}&= \langle m^p\rangle_{t,\vx}- \langle p_{12}q^{p/2}_{12}\rangle_{t,\vx},\\
\partial_x S_{N,p,\alpha}&=- \langle q_{12}\rangle_{t,\vx},\\
\partial_y S_{N,p,\alpha}&=- \langle p_{12}\rangle_{t,\vx},\\
\partial_z S_{N,p,\alpha}&=\langle m\rangle_{t,\vx}.
\end{split}
\end{equation}
\end{Lem}
The computation of the spacetime derivatives is pretty lengthy but straightforward. We report the computation of the derivatives in Appendix \ref{AA:ProofLem_derSN}.
\par\medskip
In order to derive differential identities for the expectation values of the order parameters, we need to compute the spatial derivatives of a generic function of two replicas $O(\bsigma^{(1)},\bsigma^{(2)},\btau^{(1)},\btau^{(2)})$.
\begin{Prp}\label{Prop_Seq}
Let $\bb {\underline{\sigma}}= (\bsigma^{(1)},\bsigma^{(2)})$ and $\bb {\underline{\tau}}= (\btau^{(1)},\btau^{(2)})$ be the configurations of the 2-replicated system. Then
\begin{equation}\label{dx}
\begin{split}
\partial_x \langle O(\bb {\underline{\sigma}},\bb {\underline{\tau}})\rangle_{t,\vx}&=\frac N 2 \sum_{a,b=1}^2\langle O(\bb {\underline{\sigma}},\bb {\underline{\tau}})q_{ab}\rangle_{t,\vx}-2N \sum_{a=1}^2\langle O(\bb {\underline{\sigma}},\bb {\underline{\tau}})q_{a3}\rangle_{t,\vx}-N\langle O(\bb {\underline{\sigma}},\bb {\underline{\tau}})\rangle_{t,\vx} \\&+3N\langle O(\bb {\underline{\sigma}},\bb {\underline{\tau}})q_{34}\rangle_{t,\vx},
\end{split}
\end{equation}
\begin{equation}\label{dy}
\begin{split}
\partial_y \langle O(\bb {\underline{\sigma}},\bb {\underline{\tau}})\rangle_{t,\vx}&=\frac{N}{ 2 }\sum_{a,b=1}^2\langle O(\bb {\underline{\sigma}},\bb {\underline{\tau}})p_{ab}\rangle_{t,\vx}-2N \sum_{a=1}^2\langle O(\bb {\underline{\sigma}},\bb {\underline{\tau}})p_{a3}\rangle_{t,\vx}-N\langle O(\bb {\underline{\sigma}},\bb {\underline{\tau}})p_{33}\rangle_{t,\vx}\\&+3N\langle O(\bb {\underline{\sigma}},\bb {\underline{\tau}})p_{34}\rangle_{t,\vx}.
\end{split}
\end{equation}
and
\begin{equation}\label{dz}
\partial_z \langle O(\bb {\underline{\sigma}},\bb {\underline{\tau}})\rangle_{t,\vx}=N \left(\langle O(\bb {\underline{\sigma}},\bb {\underline{\tau}})m\rangle_{t,\vx}-\langle O(\bb {\underline{\sigma}},\bb {\underline{\tau}})\rangle_{t,\vx} \langle m\rangle_{t,\vx}\right)
\end{equation}
\end{Prp}
The complete proof is given in Appendix \ref{BB:Dim_Seq}. 
\begin{Cor}
For all $l\in\mathbb{N}^+$, the following equalities holds
\begin{equation}\label{dx_pqh}
\partial_x\langle p_{12}q_{12}^l\rangle_{t,\vx}=N\left(\langle p_{12}q_{12}^{l+1}\rangle_{t,\vx}-4\langle p_{12}q_{12}^lq_{13}\rangle_{t,\vx}+3\langle p_{12}q_{12}^lq_{34}\rangle_{t,\vx}\right).
\end{equation}
\end{Cor}
\begin{proof}
The proof is a simple application of Eq. \eqref{dx} with $O(\bb {\underline{\sigma}},\bb {\underline{\tau}})=p_{12}q_{12}^l$. Thus
\begin{equation*}
\begin{split}
\partial_x\langle p_{12}q_{12}^l\rangle_{t,\vx}&=\frac N2\left(\langle p_{12}q_{12}^l\rangle_{t,\vx}+\langle p_{12}q_{12}^l q_{12}\rangle_{t,\vx}+\langle p_{12}q_{12}^lq_{21}\rangle_{t,\vx}+\langle p_{12}q_{12}^l\rangle_{t,\vx}\right)\\&-2N\left(\langle p_{12}q_{12}^lq_{13}\rangle_{t,\vx}+\langle p_{12}q_{12}^lq_{23}\rangle_{t,\vx}\right)-N\langle p_{12}q_{12}^l\rangle_{t,\vx}+3N\langle p_{12}q_{12}^lq_{34}\rangle_{t,\vx}\\&=N\langle p_{12}q_{12}^{l+1}\rangle_{t,\vx}-4N\langle p_{12}q_{12}^l q_{13}\rangle_{t,\vx}+3N\langle p_{12}q_{12}^lq_{34}\rangle_{t,\vx}.
\end{split}
\end{equation*}
\end{proof}
Also in this case, we will work under the RS assumption in order to simplify the computations by neglecting the fluctuations of the order parameters w.r.t. their expectation values.
\begin{Prp}\label{Prp:equalities}
The following equalities holds 
\begin{eqnarray}
	\langle m^{p}\rangle_{t,\vx}&=&\left(\frac 1 N\partial_z+\langle m\rangle_{t,\vx}\right)^{p-1}\langle m\rangle_{t,\vx},\label{mh}\\
	\langle p_{12}q_{12}^{p/2}\rangle_{t,\vx}&=&\left(\frac 1N\partial_x+\langle q_{12}\rangle_{t,\vx}\right)^{p/2}\langle p_{12}\rangle_{t,\vx}+R_N^{(\frac p2)}(t,\bb x),\label{pqh}
\end{eqnarray}
where $R_N^{(\frac p2)}(t,\bb x)$ collects the terms involving the fluctuations of the order parameters, and thus vanishes in the $N\to \infty$ limit and under the RS assumption.
\end{Prp}
\proof
To simplify the notation, we will drop the subscript $t,\vx$ from the quenched averages. The derivation of \eqref{mh} iterating the property \eqref{dz} with $O (\bb\sigma,\bb\tau)=m(\bb\sigma)^{p-1}$. 
Let us now observe that
\begin{equation*}
\langle p_{12}q_{12}^lq_{13}\rangle=\left\langle p_{12}q_{12}^h\Delta(q_{13})\right\rangle +\langle p_{12}q_{12}^l\rangle\langle q_{13}\rangle=\langle p_{12}q_{12}^l\rangle\langle q_{13}\rangle+R^{(1,l)}_{N}(t,\bb x),
\end{equation*}
and 
\begin{equation*}
\langle p_{12}q_{12}^lq_{34}\rangle=\left\langle p_{12}q_{12}^l\Delta(q_{34})\right\rangle +\langle p_{12}q_{12}^l\rangle\langle q_{34}\rangle=\langle p_{12}q_{12}^l\rangle\langle q_{34}\rangle+R^{(2,l)}_{N}(t,\bb x),
\end{equation*}
where $\Delta(q_{ab}):=q_{ab}-\langle q_{ab}\rangle$, and $R^{(1,l)}_N, R^{(2,l)}_N$ collect the contributions involving the fluctuations. Then, recalling eq.\eqref{dx_pqh}, we can write 
\begin{equation*}
\begin{split}
\langle p_{12}q_{12}^{l+1}\rangle&=\frac 1 N \partial_x\langle p_{12}q_{12}^{l}\rangle +4\langle p_{12}q_{12}^l\rangle\langle q_{13}\rangle-3\langle p_{12}q_{12}^l\rangle\langle q_{34}\rangle+R^{(l)}_{N}(t,\bb x)=
\\&=\frac 1 N \partial_x\langle p_{12}q_{12}^{l}\rangle +\langle p_{12}q_{12}^l\rangle\langle q_{12}\rangle+R^{(l)}_{N}(t,\bb x),
\end{split}
\end{equation*}
where $R^{(l)}_{N}(t,\bb x):=4R^{(1,l)}_{N}(t,\bb x)-3R^{(2,l)}_{N}(t,\bb x)$ and we used $\langle q_{34}\rangle=\langle q_{13}\rangle=\langle q_{12}\rangle$ following from the invariance under replica labelling. The previous equation can be thus written in the following way
\begin{equation}
\langle p_{12}q_{12}^{l+1}\rangle=\left(\frac 1 N \partial_x +\langle q_{12}\rangle \right)\langle p_{12}q_{12}^l\rangle+R^{(l)}_{N}(t,\bb x).
\end{equation}
Iterating the procedure, we get
\begin{equation}
\langle p_{12}q_{12}^{l+1}\rangle=\left(\frac 1 N \partial_x +\langle q_{12}\rangle \right)^{l+1}\langle p_{12}\rangle+R_N^{(\frac p2)}(t,\bb x),
\end{equation} 
where $R_N^{(\frac p2)}(t,\bb x)$ collects all the terms involving the rests of previous expansions (and thus vanishes in the $N\to\infty$ limit and under the RS assumption). Then, by imposing $l=p/2-1$ we get the thesis. 
\endproof
Now we can use all the information obtained to build a PDE that can describe the thermodynamics of the DAM models. Indeed, recalling the temporal derivative of the Guerra's action \eqref{dA_DHN} and using the result obtained in Prop. \ref{Prp:equalities}, we have 
\begin{equation}
\partial_t S_{N,p,\alpha}=\left(\frac 1 N\partial_z+\langle m\rangle_{t,\vx}\right)^{p-1}\langle m\rangle_{t,\vx}-\left(\frac 1 N \partial_x +\langle q_{12}\rangle_{t,\vx} \right)^{p/2}\langle p_{12}\rangle_{t,\vx}-R_N^{(\frac p 2)}(t,\bb x).
\end{equation}
Finally, taking the spatial derivatives of this expression and denoting $\bb D_N (t,\bb x)=-\bb \nabla R_N^{( p/ 2)}$, we have
\begin{equation}\label{sistN_finale}
\begin{split} -\partial_t \langle q_{12}\rangle_{t,\vx}-\partial_x\Big(\frac 1 N\partial_z+\langle m\rangle_{t,\vx}\Big)^{p-1}\langle m\rangle_{t,\vx}+\partial_x\Big(\frac 1 N\partial_x+\langle q_{12}\rangle_{t,\vx}\Big)^{p/2}\langle p_{12}\rangle_{t,\vx}&= D_{N,x},\\ 
\\
-\partial_t \langle p_{12}\rangle_{t,\vx}-\partial_y\Big(\frac 1 N\partial_z-\langle m\rangle_{t,\vx}\Big)^{p-1}\langle m\rangle_{t,\vx}+\partial_y\Big(\frac 1 N\partial_x+\langle q_{12}\rangle_{t,\vx}\Big)^{p/2}\langle p_{12}\rangle_{t,\vx}&=D_{N,y},\\
\\
 \partial_t \langle m\rangle_{t,\vx}-\partial_z\Big(\frac 1 N\partial_z+\langle m\rangle_{t,\vx}\Big)^{p-1}\langle m\rangle_{t,\vx}+ \partial_z\Big(\frac 1 N\partial_x+\langle q_{12}\rangle_{t,\vx}\Big)^{p/2}\langle p_{12}\rangle_{t,\vx}&=D_{N,z}.
\end{split} 
\end{equation}
The l.h.s. of the system of PDEs constitute the 3+1-dimensional DAM generalization of the Burgers hierarchy structure. Similarly to the Derrida model case, at finite $N$ we have a source term on the r.h.s. which vanishes in the limit $N\to\infty$ under the RS assumption of the order parameters. In this case, we can analyse the thermodynamic limit and describe the equilibrium dynamics of the model.

\begin{Thm}
The high-storage regime for the DAM models under the RS assumption in the thermodynamic limit can be described by the following system of partial differential equations:
\begin{equation}\label{sist_finale}
\begin{cases} \partial_t \bar q+\partial_x\bar m^{p}-\partial_x\bar q^{p/2} \bar p&=0\\ 
\partial_t \bar p+\partial_y\bar m^{p}-\partial_y\bar q^{p/2}\bar p&=0\\
 \partial_t \bar m-\partial_z\bar m^{p}+ \partial_z\bar q^{p/2}\bar p&=0
\end{cases} ,
\end{equation}
with the initial conditions
\begin{equation}\label{cond_i}
\begin{cases}
\bar q(0,\vx)&=\E_\eta\left[\tanh^2\left(\sqrt x \eta +\frac z 2\right)\right]\\
\bar p(0,\vx)&=\begin{cases} \frac{\alpha y}{(1+y-t_0)^2} & \mbox{if }p=2 \\ \alpha y & \mbox{if }p=2k \mbox{ with } k=2,3,\dots
\end{cases}\\
\bar m(0,\vx)&=\E_\eta\left[\tanh\left(\sqrt x \eta +\frac z 2\right)\right]
\end{cases},
\end{equation}
where $\E_\eta$ is the Gaussian average over the variable $\eta$.
\end{Thm}
\begin{proof}
First, let us call
$$
\bar m(t,\vx)=\lim_{N\to \infty}\langle m\rangle_{t,\vx},\quad \bar q(t,\vx)=\lim_{N\to \infty}\langle q_{12}\rangle_{t,\vx},\quad \bar p(t,\vx)=\lim_{N\to \infty}\langle p_{12}\rangle_{t,\vx},
$$
the expectation values of the order parameters in the thermodynamic limit. Taking $N\to\infty$ in \eqref{sistN_finale} and recalling that the source contributions $\bb D_N (t,\vx)$ vanish in this limit under the RS assumption, we arrive at the PDE system \eqref{sist_finale}. Let us now find the initial conditions \eqref{cond_i}. To do this we start calculating the interpolating partition function in $t=0$
\begin{equation*}
\begin{split}
Z_{N,p,\bb \xi,\alpha}(0,\vx)&=\sum_{\bsigma}\int d\mu(\btau)\exp\biggl(\sqrt x \sum_{i=1}^N \eta_i \sigma_i +\sqrt{
\frac y {N^{p/2-1}}} \sum_{\mu=1}^K \theta_\mu \tau_\mu -
\frac{y-t_0}{2 N^{p/2-1}}
\sum_{\mu=1}^K \tau_\mu^2 +\frac z 2   \sum_{i=1}^N\sigma_i\biggl)=\\
&=\sum_{\bsigma}\exp\Big(\sum_{i=1}^N (\sqrt x \eta_i +\tfrac 1 2 z)\sigma_i \Big)\int d\mu(\btau)\exp\biggl(\sqrt{
	\frac y {N^{p/2-1}}}  \sum_{\mu=1}^K \theta_\mu \tau_\mu -\frac{y-t_0}{2 N^{p/2-1}}\sum_{\mu=1}^K \tau_\mu^2\biggl)=\\
&=\prod_{i=1}^N2\cosh\left(\sqrt x  \eta_i +\frac z 2 \right)\prod_{\mu=1}^K\frac{1}{\sqrt{N^{1-p/2}(y-t_0)+1}}\exp\left(\frac{N^{1-p/2}y\theta_\mu^2}{2\left(N^{1-p/2}(y-t_0)+1\right)}\right).
\end{split}
\end{equation*}
By using the definition of the interpolating statistical pressure \eqref{eq:Gen_Guerra}, we see that
\begin{equation*}
\begin{split}
A_{N,p,\alpha}(0,\vx)&=\frac 1 N\sum_{i=1}^N\E\log 2\cosh\left(\sqrt x  \eta_i +\frac z 2 \right)+\frac 1 N\sum_{\mu=1}^K\E\left(\frac{N^{1-p/2}y\theta_\mu^2}{2\left(N^{1-p/2}(y-t_0)+1\right)}\right)\\& -\frac{K}{2N}\log\left(N^{1-p/2}(y-t_0)+1\right)=\\&=\E\log 2\cosh\left(\sqrt x  \eta+\frac z 2 \right)+\frac{K}{N^{p/2}}\frac{y}{2\left(N^{1-p/2}(y-t_0)+1\right)}\\& -\frac{K}{2N}\log\left(N^{1-p/2}(y-t_0)+1\right),
\end{split}
\end{equation*}
where we used the fact that the $\eta_i's$ are i.i.d. random variables and $\mathbb E \theta_\mu^2 =1$ for all $\mu=1,\dots,K$.
Now, recalling Eq. \eqref{S}, we can straightforwardly derive the initial condition for the order parameters according to \eqref{dA_DHN}. First
\begin{equation}
	\begin{split}
	 q({0,\vx})&=\lim_{N\to\infty}\bigl(-\partial_x S_{N,p,\alpha}(0,\vx)\bigl)=-\partial_x \Big[\E_\eta\log 2\cosh\left(\sqrt x  \eta+\frac z 2 \right)-x \Big]=\\
		&=\E_\eta\left[\tanh^2\left(\sqrt x \eta +\frac z 2\right)\right].
	\end{split}
\end{equation}
Analogously, we have
\begin{equation}
	\begin{split}
		p(0,\vx)&=\lim_{N\to \infty}\bigl(-\partial_y S_{N,p,\alpha}(0,\vx)\bigr)=\\&=-\lim_{N\to \infty}\partial_y\Big( \frac K{N^{p/2}}\frac{y}{1+N^{1-p/2}(y-t_0)}-\frac KN\log\big[1+N^{1-p/2}(y-t_0)\big] \Big)=\\
		&=-\lim_{N\to \infty}\left(- \frac K{N^{p-1}}\frac y{[1+N^{1-p/2}(y-t_0)]^2}\right)=
		\\
		&
		=\begin{cases}
			\frac {\alpha y}{[1+(y-t_0)]^2} \quad& \text{if } p=2\\
			\alpha y \quad& \text{if } p=2k \text{ and } k=2,3,\dots\\
		\end{cases}.
	\end{split}
\end{equation}
Finally
\begin{equation}
	\begin{split}
		m(0,\vx )=\lim_{N\to \infty}\partial_z S_{N,p,\alpha}(0,\vx)=\partial_z \E_\eta\log 2\cosh\left(\sqrt x  \eta+\frac z 2 \right)= \E_\eta \tanh\left(\sqrt x  \eta+\frac z 2 \right).
	\end{split}
\end{equation}
This concludes the proof.
\end{proof}

\begin{Prp}\label{G}
	The system of PDEs \eqref{sist_finale} can be rewritten in a non-linear wave equation as
	\begin{equation}\label{vec_equation}
		\partial _t \bb \phi+( \bb v\cdot \bb\nabla)\bb \phi=0,
	\end{equation}
where $\bb \phi:=(\bar q,\bar p,\bar m)$ is the vector of the order parameters and $\vv:=\left(-\frac p2\bar q^{p/2-1}\bar p,-\bar q^{p/2},-p\bar m^{p-1}\right)$ is the effective velocity.
\end{Prp}
\begin{proof}
We prove the equation \eqref{vec_equation} for the first component, as the others follow accordingly. Let us define the function $G(\bb\phi)=\bar m^p-\bar q^{p/2}\bar p$, so that the PDE for the order parameter $\bar q$ can be rewritten as
$$
\partial_t \bar q +\partial_x G(\bb \phi)=0.
$$
The $x$-derivative of the function $G$ is straightforwardly computed:
$$
\partial _x G(\bb\phi)=\partial_x (\bar m^p-\bar q^{p/2}\bar p)=p \bar m^{p-1}\partial_x \bar m - \frac p2 \bar q^{p/2-1} \bar p \ \partial_x \bar q-\bar q^{p/2}\partial_x \bar p.
$$
Now
$$
\partial_ x \bar p =\partial_x\left(-\lim_{N\to \infty}\partial_y S_{N,p,\alpha}(t,\vx)\right)=\partial_y \left(-\lim_{N\to \infty}\partial_x S_{N,p,\alpha}(t,\vx)\right)=\partial_y \bar q.
$$
In the same way
$$
\partial_ x \bar m= \partial_x\left(\lim_{N\to \infty}\partial_z S_{N,p,\alpha}(t,\vx)\right)= \partial_z\left(\lim_{N\to \infty}\partial_x S_{N,p,\alpha}(t,\vx)\right)=-\partial _z \bar q.
$$
With these results, we have
$$
\partial_x G(\bb\phi)=-\frac p2 \bar q^{p/2-1}\bar p \ \partial _x \bar q-\bar q^{p/2}\partial_y \bar q-p \bar m^{p-1}\partial_z \bar q= (v_x \partial_x +v_y \partial_y +v_z \partial_z)\bar q= (\bb v\cdot \bb \nabla)\bar q.
$$
This leads to
$$
\partial_t \bar q +(\bb v \cdot \bb \nabla )\bar q=0.
$$
\end{proof}

\begin{Prp}\label{pr_final}
The equilibrium dynamics of the DAM models is given by the set of self-consistency equations
\begin{equation}\label{auto_p}
\begin{split}
&\bar q=\E\left[\tanh^2\left(\sqrt{\beta'\frac p2\bar q^{p/2-1}\bar p}\eta+\beta'\frac p2\bar m^{p-1}\right)\right],\\
&\bar m=\E\left[\tanh\left(\sqrt{\beta'\frac p2\bar q^{p/2-1}\bar p}\eta+\beta'\frac p2\bar m^{p-1}\right)\right],\\
&\bar p=\begin{cases} \frac{\alpha\beta\bar q}{[1-\beta(1-\bar q)]^2}, & \mbox{if }p=2 \\ \alpha\beta'\bar q^{p/2}  & \mbox{if }p=2k \mbox{ with } k=2,3,\dots
\end{cases}\\
\end{split}
\end{equation}
where $\beta':=\frac{2\beta}{p!}$.
\end{Prp}
\proof
In order to proof our assertion, we use the vector PDE \eqref{vec_equation}, whose solution can be given in implicit form as
$$
\bb\phi (t,\vx)= \bb\phi_0 (\vx-\vv t),
$$
where $\bb\phi_0(\bb x)$ is the initial profile given by the conditions \eqref{cond_i}. For the first component, we have
\begin{equation}
	\begin{split}
		\bar q(t,\vx)&= \phi_{0,x}(\bb x -\bb v t)=\E_\eta\tanh^2\left(\sqrt{ x-v_x t}\eta +\frac {z-v_z t} 2\right)=\\
		&=\E_\eta\tanh^2\left(\sqrt{ x+\frac p2 \bar q^{p/2-1}\bar p t}\eta +\frac {z+p \bar m^{p-1}t} 2\right).
	\end{split}
\end{equation}
Analogously,
\begin{equation}
	\begin{split}
		\bar m(t,\vx)& = \phi_{0,z}(\bb x -\bb v t)=\E_\eta\tanh\left(\sqrt{ x+\frac p2 \bar q^{p/2-1}\bar p t}\eta +\frac {z+p \bar m^{p-1}t} 2\right).
	\end{split}
\end{equation}
Finally, if $p=2$, we have
\begin{equation}
	\begin{split}
		\bar p(t,\vx )&=\phi_{0,y}(\vx -\vv t)=\frac{\alpha (y-v_y t)}{(1+y-v_y t-t_0)^2}=
		\frac{\alpha (y+\bar q t)}{(1+y+\bar q t-t_0)^2},
	\end{split}
\end{equation}
while, for $p=2k$ with $k\ge2$, the same order parameter satisfies the self-consistency equation
\begin{equation}
	\begin{split}
		\bar p(t,\vx )&=\phi_{0,y}(\vx -\vv t)={\alpha (y-v_y t)}={\alpha (y+\bar q t)}.
	\end{split}
\end{equation}
Recalling that the thermodynamics of the DAM models is reproduced when $t=t_0=\frac{2\beta}{p!}$ and $\vx =\bb 0$, we easily get the thesis.
\endproof

\begin{Rmk}
We stress that the self-consistency equations in Prop. \ref{pr_final} are in agreement with those obtained by Gardner in \cite{gardner1987} by means of replica trick.
\end{Rmk}
%

\begin{figure}
	\centering
\includegraphics[width=0.49\textwidth]{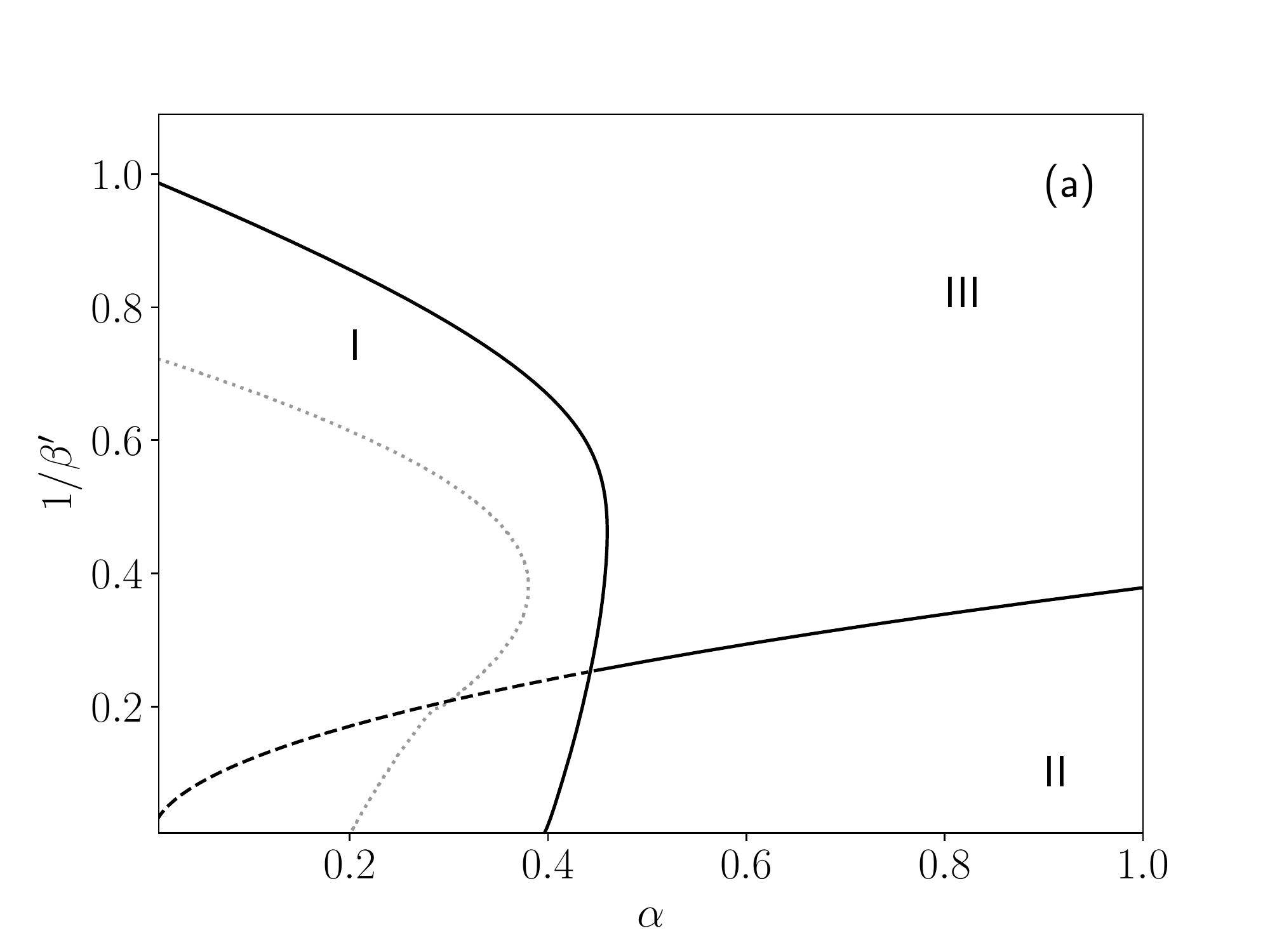}
\includegraphics[width=0.49\textwidth]{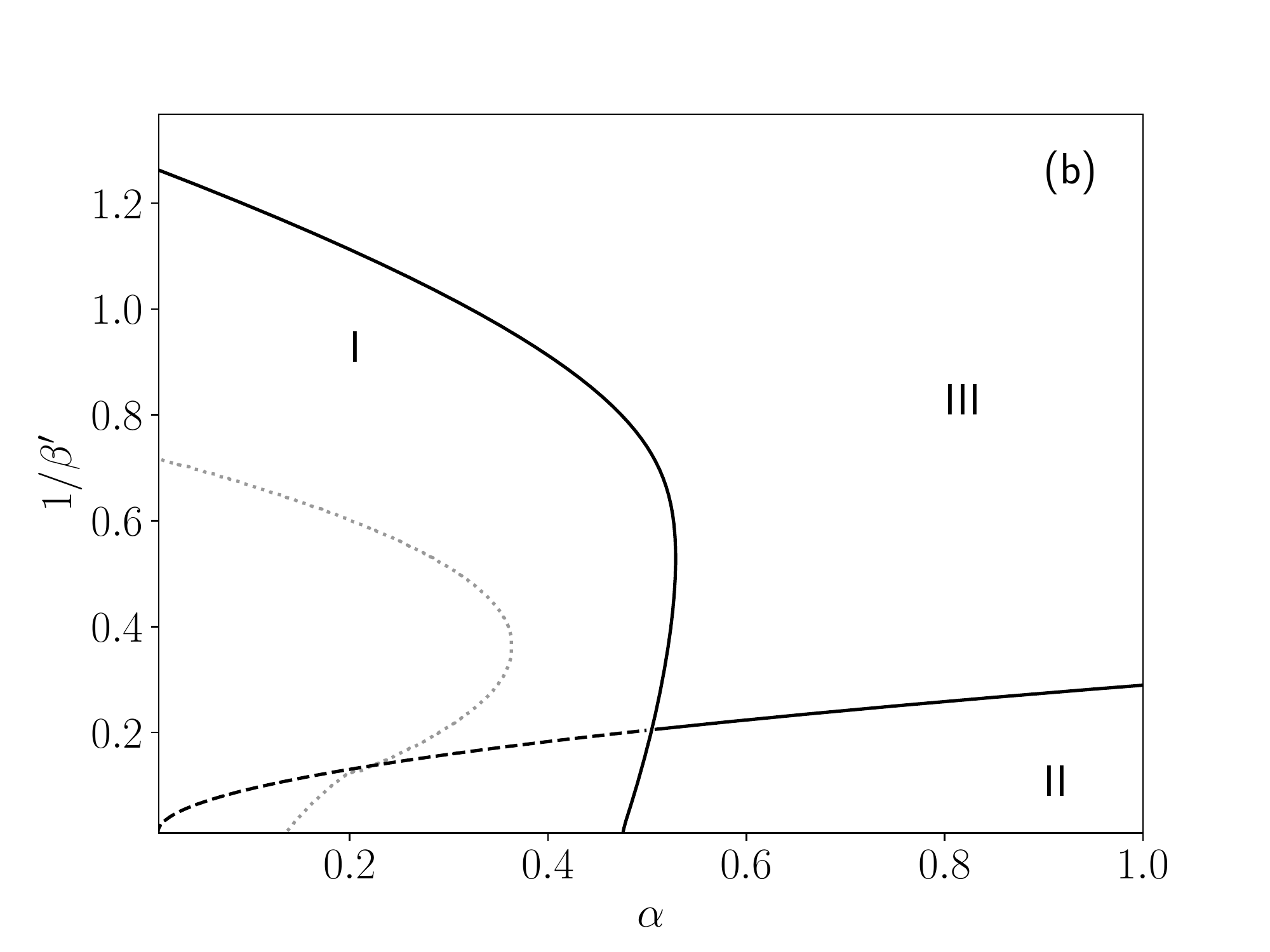}\\
\includegraphics[width=0.49\textwidth]{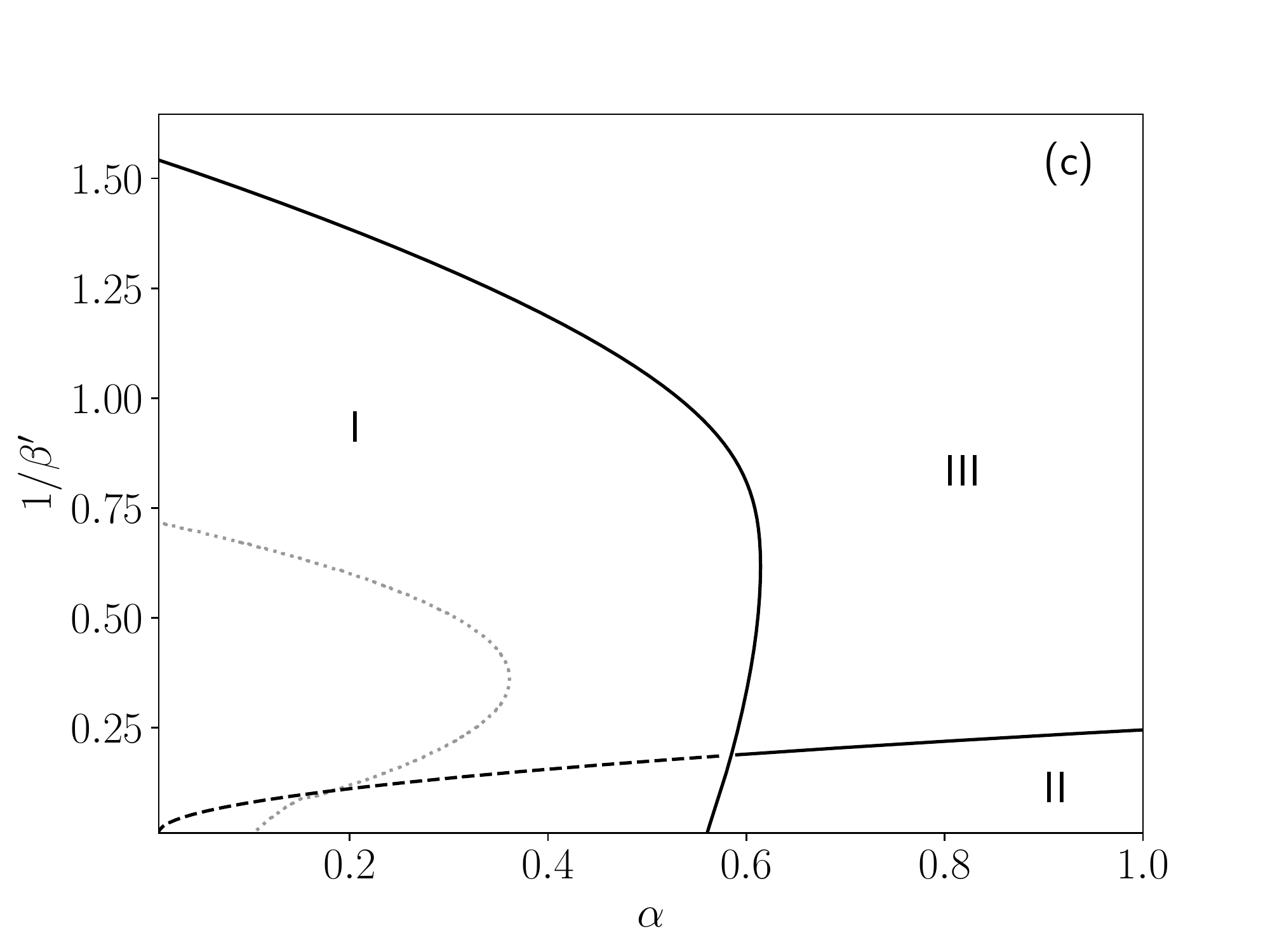}
\includegraphics[width=0.49\textwidth]{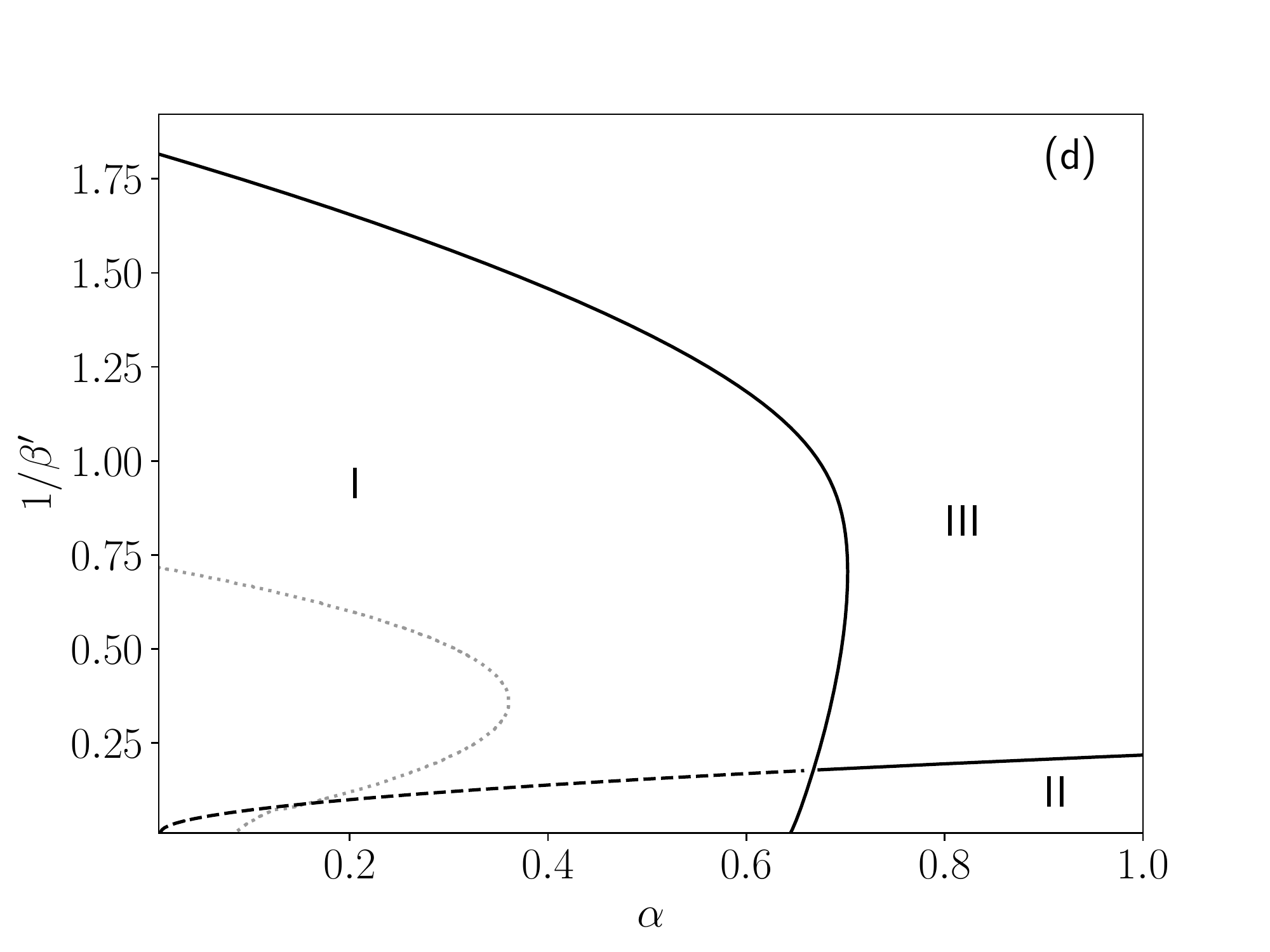}
	\caption{Phase diagram for $p=4$ (panel a), $p=6$ (panel b), $p=8$ (panel c) and $p=10$ (panel d), in the space of the control parameters $\alpha$ and $\beta' = 2\beta/p!$. The region I, delimited by the solid black line, is the retrieval region, while region II and region III are respectively the spin-glass and the ergodic phases. The dashed line is the prolongation of the spin-glass phase inside the retrieval region. Finally, the lighter dotted line inside the region I identifies the portion of the parameter space in which the retrieval states are global minima for the free energy. Notice that the indentation that can be seen in the transition line delimiting the retrieval phase is a spurious effects due to the RS approximation \cite{albanese2021}. 
}\label{p_var.}
\end{figure}

By solving these self-consistency equations \eqref{auto_p} numerically by a fixed-point method for a given $p$ and tuning the parameters $T$ and $\alpha$, we obtain the phase diagrams shown in Fig.~\ref{p_var.}. As expected, the diagrams exhibit the existence of three different regions: 
\begin{itemize}
\item For high levels of noise $T$, no matter the value of storage $\alpha$, the only stable solution is given by $\bar m = 0$, $\bar q = 0$, thus the system is ergodic (III);
\item At lower temperatures and with relatively high load, the system exhibits spin-glass behaviors (II), and the solution is characterized by $\bar m = 0$ and $\bar q\neq0$;
\item For relatively small values of $\alpha$ and $T$, we have $\bar m,\bar q \neq 0$ and the system is located in the retrieval phase (I). In this situation, the system behaves as an associative neural network performing spontaneously pattern recognition. In particular, we can see that the retrieval region observed for values of $T$ and $\alpha$ relatively small can be further split in a pure retrieval region , where pure states are global minima for the free energy, and in a mixed retrieval region, where pure states are local minima, yet their attraction basin is large enough for the system to end there if properly stimulated. 
\end{itemize}
Thus, by increasing $p$ we need to afford higher costs in terms of resources, since the number of connection weights to be properly set grows as $\binom{N}{p}$, but we also have a reward on a coarse scale, since the number of storable patterns grows as $K \sim N^{p-1}$, as well as on a fine scale, since the critical load $\alpha_c$ is also increasing with $p$.

\section{Conclusions}
In this work we focused on DAMs, that are neural-network models widely used for pattern recognition tasks and characterized by high-order (higher than quadratic) interactions between the constituting neurons. Extensive empirical evidence has shown that these models significantly outperform non-dense networks (displaying only quadratic interactions), especially as for the ability to correctly recognize adversarial or extremely noisy examples \cite{Krotov2016DenseAM,Krotov2018DAMS,AABCF-PRL2020,AD-EPJ2020} hence making these models particularly suitable for detecting and cope with malicious attacks.
From the theoretical side, results are sparse and mainly based on the possibility of recasting these networks as spin-glass-like models with spins interacting $p$-wisely; these models can, in turn, be effectively faced by tools stemming from the statistical mechanics of disordered systems (\emph{e.g.}, \cite{Bovier-JPA1994,AABF-NN2020}). Here, we pave this way and develop analytical techniques for their investigation. 
More precisely, we translate the original statistical-mechanical problem into an analytical-mechanical one where control parameters play as spacetime coordinates, the free-energy plays as an action and the macroscopic observables that assess whether the system can be used for pattern recognition tasks play as effective velocities and are shown to fulfil a set of nonlinear partial differential equations. In this framework, transitions from different regimes ({\it e.g.}, from a region in the control parameter space where the system performs correctly and another one where information processing capabilities are lost) appear as the emergence of shock waves.

A main advantage of this route is that it allows for rigorous investigations in a field where most knowledge is based on (pseudo) heuristic approaches, with a wide set of already available methods and strategies to rely upon. Further, by bridging two different perspectives, the statistical-mechanics and the analytical one, we anticipate a cross-fertilization that may lead to a deeper comprehension of the system subtle mechanisms and ultimately progress in the development of a complete theory for (deep) machine learning.

\section*{Acknowledgments}
E.A. and A.F. would like to thank A. Moro for useful discussions.
The authors acknowledge Sapienza University of Rome for financial support (RM120172B8066CB0, RM12117A8590B3FA, AR12117A623B0114).

\appendix
\section{Proof of Lemma \ref{Prp:Sderivs_pSK}}\label{AA:Sderivs_pSK}
\begin{proof}
	First of all, we compute the temporal derivative:
	\begin{equation}
	\begin{split}
		\partial _t A_{N,p} (t,x)&=\frac1N \mathbb E_{\bb J} \frac1{Z_{N,p,\bb J}(t,x)}\sum_{\bb \sigma}\sqrt{\frac{p!}{2 N^{p-1}}} \frac1 {2\sqrt t}	\sum_{1\le i_1<\dots<i_p \le N}J_{i_1\dots i_p}\sigma_{i_1}\dots \sigma_{i_p} B_{N,p,\bb J} (t,x)=\\
		&=\frac1N \sqrt{\frac{p!}{2 N^{p-1}}}{\frac1{2\sqrt t}}\sum_{1\le i_1<\dots<i_p \le N}\mathbb E _{\bb J} J_{i_1\dots i_p}\omega_{t,x}(\sigma_{i_1}\dots \sigma_{i_p}).
	\end{split}
	\end{equation}
	Here, we can use the Wick-Isserlis theorem for normally distributed random variables, ensuring that $\mathbb E _{\bb J} J_l f(\bb J)=\mathbb E _{\bb J} \partial_{J_l} f(\bb J) $ for each function $f$ of the quenched disorder $\bb J$. Thus
	\begin{equation}
		\begin{split}
			\partial _t A_{N,p}(t,x)&=\frac1N \sqrt{\frac{p!}{2 N^{p-1}}}{\frac1{2\sqrt t}}\sum_{1\le i_1<\dots<i_p \le N}\mathbb E _{\bb J}\partial_{ J_{i_1\dots i_p}}\omega_{t,x}(\sigma_{i_1}\dots \sigma_{i_p})=\\
			&=\frac1N \sqrt{\frac{p!}{2 N^{p-1}}}{\frac1{2\sqrt t}} \sqrt{\frac{t p!}{2N^{p-1}}}\sum_{1\le i_1<\dots<i_p \le N}\mathbb E _{\bb J}(\omega_{t,x}(1)-\omega_{t,x}(\sigma_{i_1}\dots \sigma_{i_p})^2)=\\
			&=\frac{p!}{4 N^p}\sum_{1\le i_1<\dots<i_p \le N}(1-\mathbb E_{\bb J}\omega_{t,x}(\sigma_{i_1}\dots \sigma_{i_p})^2).
		\end{split}
	\end{equation}
	The non-trivial contribution in the round brackets in the last equality can be expressed in terms of the overlap order parameter. Indeed
	\begin{equation}
		\begin{split}
			\mathbb E_{\bb J}\omega_{t,x}(\sigma_{i_1}\dots \sigma_{i_p})^2 &=\mathbb E_{\bb J}\omega_{t,x}^{(1)}(\sigma_{i_1}^{(1)}\dots \sigma_{i_p}^{(1)})\ \omega_{t,x}^{(2)}(\sigma_{i_1}^{(2)}\dots \sigma_{i_p}^{(2)})=\\
			&= \mathbb E_{\bb J} \Omega^{(2)}_{t,x}(\sigma_{i_1}^{(1)} \sigma_{i_1}^{(2)} \dots \sigma_{i_p}^{(1)} \sigma_{i_p}^{(2)})=\\
			&= \langle \sigma_{i_1}^{(1)} \sigma_{i_1}^{(2)} \dots \sigma_{i_p}^{(1)} \sigma_{i_p}^{(2)}\rangle_{t,x},
		\end{split}
	\end{equation}
	where $1,2$ are the replica indices. Now, using Rem. \ref{Rmk:sum}, in the thermodynamic the following equality holds:
	\begin{equation}
		\begin{split}
			\partial _t A_{N,p}(t,x)&=\frac1{4 N^p}\sum_{i_1,\dots,i_p =1}^N(1-\langle \sigma_{i_1}^{(1)} \sigma_{i_1}^{(2)} \dots \sigma_{i_p}^{(1)} \sigma_{i_p}^{(2)}\rangle_{t,x})=\\
			&=\frac14(1-\langle\bigl(\frac1N\sum_{i=1}^N\sigma_i^{(1)}\sigma_i^{(2)}\bigr) ^p\rangle_{t,x})=\frac14 (1-\langle q_{12 }^p\rangle_{t,x}).
		\end{split}
	\end{equation}
	Concerning the spatial derivative, we proceed in the same way:
	\begin{equation}
		\begin{split}
			\partial_x A_{N,p}(t,x)&= \frac1N \mathbb E_{\bb J}\frac1{Z_{{N,p,\bb J}}(t,x)}\sum_{\bb \sigma} \frac1{2\sqrt x}\sum_{i=1}^N J_i \sigma_i B_{N,p,\bb J} ( t,x)=\\
			&=\frac1{2N \sqrt x} \sum_{i=1}^N\mathbb E_{\bb J} J_i \omega_{t,x}(\sigma_i)=\frac1{2N \sqrt x} \sum_{i=1}^N\mathbb E_{\bb J} \partial_{J_i }\omega_{t,x}(\sigma_i)=\\
			&=\frac1{2N}\sum_{i=1}^N (1-\mathbb E_{\bb J}\omega_{t,x}(\sigma_i)^2).
		\end{split}
	\end{equation}
	Also in this case, we express $\mathbb E_{\bb J}\omega_{t,x}(\sigma_i)^2=\langle \sigma_i ^{(1)} \sigma_i ^{(2)}\rangle_{t,x}$, thus
	\begin{equation}
		\begin{split}
			\partial_x A_{N,p}(t,x)&=\frac12 (1-\langle q_{12}\rangle_{t,x}).
		\end{split}
	\end{equation}
By simply exploiting Def. \ref{Def:GuerraAction_pSK} and the Rem.\ref{Rmk:sum} we get the thesis.
\end{proof}

\section{Particular cases of low storage DAMs}\label{BB:particular cases}
In this Appendix, having clear the equations describing the general case of the DAM models in the low storage regime, we will study two special cases namely the standard case where $p = 2$ and the more complex case with $p = 3$. In particular, we will observe that these two cases can be described by the Burgers and Sharma-Tasso-Olver equations in a $(K+1)$-dimensional space, respectively. To start this study, however, we first need the following definition and Lemma. 

\begin{Lem}\label{prop_com}
	For all $\mu,\nu=1,\dots,K$,  we have 
	\begin{enumerate}
		\item $[D_\mu,D_\nu]=0 $;
		\item $[\partial_\mu, D_\nu ^s]=N[D_\nu ^s,u_\mu] $ $\forall s>0,$
	\end{enumerate}
where $[\cdot,\cdot]$ is the usual commutator.
\end{Lem}
\begin{proof}
	The proof of the statement $1.$ works by direct computation. Indeed:
	\begin{equation*}
		\begin{split}
			[D_\mu,D_\nu]&=\Big[\frac1N \partial_\mu+u_\mu,\frac1N\partial_\nu+u_\nu\Big]=\\&
			=\frac1{N^2} [\partial_\mu,\partial_\nu]+\frac1N [\partial_\mu,u_\nu]+\frac1N [u_\mu,\partial_\nu]+[u_\mu,u_\nu]=\\
			&=\frac1N(\partial_\mu u_\nu -\partial_\nu u_\mu).
		\end{split}
	\end{equation*}
	Since the field $u_\mu(t,\vx)$ is conservative, {\it i.e.} $u_\mu(t,\vx)=\partial_\mu A_{N,p,\bb\xi,K} (t,\vx)$, we have
	$$
	\partial_\mu u_\nu -\partial_\nu u_\mu= (\partial _\mu \partial_\nu- \partial_\nu \partial_\mu )A_{N,p,\bb\xi,K}(t,\vx)=0,
	$$
	meaning that $[D_\mu,D_\nu]=0.$
	Let as now prove the property $2.$ In this case the proof works by exploiting the property $[D_\mu,D_\nu]=0$ (from which it follows also $[D_\mu,D_\nu^s]=0$) and the definition $D_\mu =\frac1N\partial_\mu+u_\mu$. Indeed
	$$
	0=[D_\mu,D_\nu^s]= \frac1N [\partial_\mu, D^s_\nu]+[u_\mu, D_\nu ^s].
	$$
	By rearranging the equality we easily get the thesis.
\end{proof}

From the previous lemma we can then prove the following two propositions

\begin{Prp}
	In the $p=2$ case with a generic finite $K$, the evolutive equations \eqref{eq:dense_guerra_9} reduces to the multidimensional Burgers equation.
\end{Prp}
\begin{proof}
	In the $p=2$ case, the equations \eqref{eq:dense_guerra_9} reduces to
	$$
	\partial_t u_\mu = -\sum_{\nu=1}^K\partial_\mu D_\nu u_\nu.
	$$
	Now, using the second claim of Lemma \ref{prop_com} with $s=1$, we have 
	$$
	\partial_\mu D_\nu u_\nu=\Big(D_\nu \partial_\mu+N[D_\nu,u_\mu] \Big) u_\nu.
	$$
	Recalling the definition of the $D$ operator, we have
	\begin{eqnarray*}
		D_\nu \partial_\mu u_\nu &=& \Big( \frac1N \partial_\nu +u_\nu\Big)\partial_\mu u_\nu =\frac1N \partial_\nu \partial_\mu u_\nu+u_\nu \partial_\mu u_\nu, \\
		\big[D_\nu,u_\mu\big]&=&\big[\frac1N\partial_\nu +u_\nu,u_\mu\big]=\frac1N\big[\partial_\nu ,u_\mu\big]=\frac1N \partial_\nu u_\mu.
	\end{eqnarray*}
	Thus
	$$
	\partial_t u_\mu = -\sum_{\nu=1}^K \Big(\frac1N \partial_\nu \partial_\mu u_\nu+u_\nu \partial_\mu u_\nu+u_\nu\partial_\nu u_\mu \Big).
	$$
	But now
	\begin{eqnarray*}
		\partial_\nu \partial_\mu u_\nu &=&\partial_\nu \partial_\mu \partial_\nu A_{N,p=2,\bb\xi,K}(t,\vx)=\partial_\nu ^2 u_\mu,\\
		u_\nu \partial_\mu u_\nu&=&u_\nu \partial_\mu \partial_\nu A_{N,p=2,\bb\xi,K}(t,\vx)=u_\nu\partial_\nu u_\mu.
	\end{eqnarray*}
	Using these results, we can rewrite the equation as
	$$
	\partial_t u_\mu+ \sum_{\nu=1}^K \Big(\frac1N \partial_\nu^2 u_\mu+2u_\nu \partial_\nu u_\mu \Big)=0,
	$$
	or, in vector form
	$$
	\partial _t \bb u+\frac1N \bb\nabla^2 \bb u+2(\bb u\cdot \bb \nabla) \bb u=0,
	$$
	which is precisely the Burgers equation in $K+1$ spacetime.
\end{proof}

\begin{Prp}
	In the $p=3$ case (and generic $K$), the evolutive equations \eqref{eq:dense_guerra_9} reduces to the multidimensional Sharma-Tasso-Olver (STO) equation \cite{olver1977evolution, tasso1976cole}.
\end{Prp}

\begin{proof}
	In the $p=3$ case, the equations \eqref{eq:dense_guerra_9} reduces to
	$$\partial_t u_\mu = -\sum_{\nu=1}^K\partial_\mu D^2_\nu u_\nu.$$
	Recalling that $D_\nu:=\left(\frac 1 N \partial_\nu+u_\nu\right)$ we have
	\begin{equation*}
		\begin{split}
			\partial_\mu D^2_\nu u_\nu&= \partial_\mu\left(\frac 1 N \partial_\nu+u_\nu\right)\left(\frac 1 N \partial_\nu+u_\nu\right)u_\nu\\&=\partial_\mu \left[\frac{1}{N^2}\partial_\nu^2 u_\nu +\frac{1}{N}\partial_\nu(u^2_\nu)+\frac{1}{N}u_\nu\partial_\nu u_\nu+u^3_\nu\right]\\&=\partial_\mu \left[\frac{1}{N^2}\partial_\nu^2 u_\nu +\frac{3}{N}u_\nu\partial_\nu u_\nu+u^3_\nu\right].
		\end{split}
	\end{equation*}
	Performing the derivative with respect to the $\mu$-th component we therefore get
	\begin{equation}
		\partial_\mu D^2_\nu u_\nu =\frac{1}{N^2}\partial_\mu\partial_\nu^2 u_\nu+\frac{3}{N} \left(\partial_\mu u_\nu \partial_\nu u_\nu +u_\nu \partial_\mu \partial_\nu u_\nu \right)+3 u_\nu^2 \partial_\mu u_\nu.
	\end{equation}
	Now, recalling that $u_\mu(t,\bb x):=\omega_{t,\bb x}(m_\nu(\bsigma))$ and \eqref{eq:dense_guerra_4}, we have
	\begin{equation*}
		\begin{split}
			 \partial_\mu\partial_\nu^2 u_\nu&=\partial_\mu\partial_\nu^3 A_{N,p=3,\bb\xi,K}(t,\bb x)=\partial_\nu^3 u_\mu\\
			 \partial_\mu u_\nu&=\partial_\mu\partial_\nu A_{N,p=3,\bb\xi,K}(t,\bb x)=\partial_\nu u_\mu\\
			 \partial_\mu\partial_\nu u_\nu&=\partial_\mu\partial_\nu^2 A_{N,p=3,\bb\xi,K}(t,\bb x)=\partial_\nu^2 u_\mu.
		\end{split}
	\end{equation*}
	Using these results, we can rewrite the equation as
	$$
	\partial_t u_\mu=- \sum_{\nu=1}^K \left(\frac{1}{N^2} \partial^3_\nu u_\mu+\frac3N(\partial_\nu u_\mu)^2 +\frac3N u_\nu\partial^2_\nu u_\mu+3 u_\nu^2 \partial_\nu u_\mu\right).
	$$
\end{proof}

\section{Proof of Lemma \ref{derSN}}\label{AA:ProofLem_derSN}
\proof
We prove the equality for the $t$-derivative of the Guerra's action, as the others follow with similar calculations. To do this, we first compute the temporal derivative of the interpolating statistical pressure:
\begin{equation*}
\begin{split}
\partial_t A_{N,p,\alpha}&=\frac 1 2\E\omega_{t,\bb x}(m^p)+\frac{1}{2N\sqrt{tN^{p-1}}}\sum_{\mu} \sum_{i_1,\dots,i_{p/2}} \E\, \xi^\mu_{i_1}\dots\xi^\mu_{i_{p/2}}\omega_{t,\bb x}(\sigma_{i_1}\dots\sigma_{i_{p/2}}\tau_\mu)-\frac12\E\omega_{t,\bb x}(p_{11}).
\end{split}
\end{equation*}
Since non-retrieved patterns constitutes a noise contribution to the system dynamics, we can assume - with standard arguments about the universality of noise \cite{Genovese_2012,Agliari_2019} - that the whole product $\xi^\mu_{i_1}\dots \xi^\mu_{i_{p/2}}$ is Gaussian-distributed as long as $N\to\infty$ and $K\to\infty$. Thus, we can apply the Wick-Isserlis theorem on the second contribution to get
\begin{equation*}
\begin{split}
\partial_t A_{N,p,\alpha}&=\frac 1 2\langle m^p\rangle_{t,\vx}+\frac{1}{2{N^{p}}}\sum_{\mu} \sum_{i_1,\dots,i_{p/2}} \E\left[\omega_{t,\bb x}(\tau_\mu^2)-\omega^2_{t,\bb x}(\sigma_{i_1}\dots\sigma_{i_{p/2}}\tau_\mu)\right]-\frac12\langle p_{11}\rangle_{t,\vx}=\\&
=\frac12 \langle m^p\rangle_{t,\vx}+\frac1{2N^p}(N^p\langle p_{11}\rangle_{t,\vx}-N^p\langle q_{12}^{p/2}p_{12}\rangle_{t,\vx})-\frac12\langle p_{11}\rangle_{t,\vx}=
\\&=\frac 1 2\langle m^p\rangle_{t,\vx}-\frac12\langle p_{12}q_{12}^{p/2}\rangle_{t,\vx},
\end{split}
\end{equation*}
where we used the definitions of the overlap order parameters \eqref{overlap_q} and \eqref{overlap_p}.
Recalling that $S_{N,p,\alpha}(t,\vx)=2A_{N,p,\alpha} (t,\vx)-x$, we finally get the result. 
\endproof

\section{Proof of Proposition \ref{Prop_Seq}}\label{BB:Dim_Seq}
\proof
We only prove the equation \eqref{dx}, the other one can be obtained in an analogous way. We will denote for simplicity of notation $\langle \cdot\rangle_{t,\vx}$ with $\langle \cdot\rangle$. Thus
\begin{equation}\label{dxF}
\begin{split}
\partial_x \langle O(\bb {\underline{\sigma}},\bb {\underline{\tau}})\rangle&=\frac{1}{2\sqrt x}\sum_{i=1}^N\sum_{a=1}^2\E_\eta\eta_i \Omega^{(2)}\left(O(\bb {\underline{\sigma}},\bb {\underline{\tau}})\sigma_i^{(a)}\right)-\frac{1}{\sqrt x}\sum_{i=1}^N\E\eta_i\Omega^{(3)}\left(O(\bb {\underline{\sigma}},\bb {\underline{\tau}})\sigma_i^{(3)}\right)=\\&=\frac{1}{2\sqrt x}\sum_{i=1}^N\sum_{a=1}^2\E_\eta\partial_{\eta_i} \Omega^{(2)}\left(O(\bb {\underline{\sigma}},\bb {\underline{\tau}})\sigma_i^{(a)}\right)-\frac{1}{\sqrt x}\sum_{i=1}^N\E\partial_{\eta_i}\Omega^{(3)}\left(O(\bb {\underline{\sigma}},\bb {\underline{\tau}})\sigma_i^{(3)}\right),
\end{split}
\end{equation}
where in the last line we used the Wick-Isserlis theorem. Now, it is simple to see that 
\begin{equation}\label{domega2}
\partial_{\eta_i} \Omega^{(2)}\left(O(\bb {\underline{\sigma}},\bb {\underline{\tau}})\sigma_i^{(a)}\right)=\sqrt x\left[\sum_{b=1}^2\Omega^{(2)}\left(O(\bb {\underline{\sigma}},\bb {\underline{\tau}})\sigma_i^{(a)}\sigma_i^{(b)}\right)-2\Omega^{(3)}\left(O(\bb {\underline{\sigma}},\bb {\underline{\tau}})\sigma_i^{(a)}\sigma_i^{(3)}\right)\right],
\end{equation}
and 
\begin{equation}\label{domega3}
\partial_{\eta_i} \Omega^{(3)}\left(O(\bb {\underline{\sigma}},\bb {\underline{\tau}})\sigma_i^{(3)}\right)=\sqrt x\left[\sum_{b=1}^3\Omega^{(3)}\left(O(\bb {\underline{\sigma}},\bb {\underline{\tau}})\sigma_i^{(3)}\sigma_i^{(b)}\right)-3\Omega^{(4)}\left(O(\bb {\underline{\sigma}},\bb {\underline{\tau}})\sigma_i^{(3)}\sigma_i^{(4)}\right)\right].
\end{equation}
By substituting \label{domega2} and \label{domega3} into \label{dxF} we get 
\begin{equation*}
\begin{split}
\partial_x \langle O(\bb {\underline{\sigma}},\bb {\underline{\tau}})\rangle&=\frac 1 2 \sum_{i=1}^N\sum_{a,b=1}^2\E\Omega^{(2)}\left(O(\bb {\underline{\sigma}},\bb {\underline{\tau}})\sigma_i^{(a)}\sigma_i^{(b)}\right)-\sum_{i=1}^N\sum_{b=1}^2\E\Omega^{(3)}\left(O(\bb {\underline{\sigma}},\bb {\underline{\tau}})\sigma_i^{(3)}\sigma_i^{(a)}\right)\\&-\sum_{i=1}^N\sum_{a=1}^3\E\Omega^{(3)}\left(O(\bb {\underline{\sigma}},\bb {\underline{\tau}})\sigma_i^{(3)}\sigma_i^{(a)}\right)+3\sum_{i=1}^N\E\Omega^{(4)}\left(O(\bb {\underline{\sigma}},\bb {\underline{\tau}})\sigma_i^{(3)}\sigma_i^{(4)}\right).
\end{split}
\end{equation*}
Recalling that $q_{ab}=\frac 1 N\sum_i\sigma_i^{(a)}\sigma_i^{(b)}$ we can write 
\begin{equation}
\begin{split}
\partial_x \langle O(\bb {\underline{\sigma}},\bb {\underline{\tau}})\rangle&=\frac N 2 \sum_{a,b=1}^2\langle O(\bb {\underline{\sigma}},\bb {\underline{\tau}}) q_{ab}\rangle-\sum_{a=1}^2\langle O(\bb {\underline{\sigma}},\bb {\underline{\tau}}) q_{a3}\rangle-\sum_{a=1}^3\langle O(\bb {\underline{\sigma}},\bb {\underline{\tau}}) q_{3a}\rangle+3N\langle O(\bb {\underline{\sigma}},\bb {\underline{\tau}}) q_{34}\rangle=
\\&= \frac N 2 \sum_{a,b=1}^2\langle O(\bb {\underline{\sigma}},\bb {\underline{\tau}}) q_{ab}\rangle-2N \sum_{a=1}^2\langle O(\bb {\underline{\sigma}},\bb {\underline{\tau}})q_{a3}\rangle-N\langle O(\bb {\underline{\sigma}},\bb {\underline{\tau}})\rangle +3N\langle O(\bb {\underline{\sigma}},\bb {\underline{\tau}})q_{34}\rangle
\end{split}
\end{equation}
thus obtaining Eq. \eqref{dx}.
\endproof

\bibliographystyle{ieeetr}
\bibliography{bib_burger}

\begin{thebibliography}{10}

\bibitem{LeCun2015}
Y.~LeCun, Y.~Bengio, and G.~Hinton, ``Deep learning,'' {\em Nature}, vol.~521,
  no.~7553, pp.~436--444, 2015.

\bibitem{Schmidhuber2015}
J.~Schmidhuber, ``Deep learning in neural networks: an overview,'' {\em Neural
  {N}etworks}, vol.~61, pp.~85--117, 2015.

\bibitem{mezard1988spin}
M.~M{\'e}zard, G.~Parisi, and M.~A. Virasoro, {\em Spin glass theory and
  beyond: An Introduction to the Replica Method and Its Applications}, vol.~9.
\newblock World {S}cientific {P}ublishing {C}ompany, 1987.

\bibitem{sherrington1975}
D.~Sherrington and S.~Kirkpatrick, ``Solvable model of a spin-glass,'' {\em
  Physical {R}eview {L}etters}, vol.~35, no.~26, p.~1792, 1975.

\bibitem{Parisi1979rsb1}
G.~Parisi, ``Toward a mean field theory for spin glasses,'' {\em Physics
  {L}etters {A}}, vol.~73, no.~3, pp.~203--205, 1979.

\bibitem{Parisi1979rsb2}
G.~Parisi, ``Infinite number of order parameters for spin-glasses,'' {\em
  Physical {R}eview {L}etters}, vol.~43, no.~23, p.~1754, 1979.

\bibitem{Parisi1980rsb3}
G.~Parisi, ``The order parameter for spin glasses: a function on the interval
  0-1,'' {\em Journal of {P}hysics {A}: {M}athematical and {G}eneral}, vol.~13,
  no.~3, p.~1101, 1980.

\bibitem{Parisi1980sequence}
G.~Parisi, ``A sequence of approximated solutions to the {SK} model for spin
  glasses,'' {\em Journal of {P}hysics {A}: {M}athematical and {G}eneral},
  vol.~13, no.~4, p.~L115, 1980.

\bibitem{mezard1984replica}
M.~M{\'e}zard, G.~Parisi, N.~Sourlas, G.~Toulouse, and M.~Virasoro, ``Replica
  symmetry breaking and the nature of the spin glass phase,'' {\em Journal de
  {P}hysique}, vol.~45, no.~5, pp.~843--854, 1984.

\bibitem{guerra2003broken}
F.~Guerra, ``Broken replica symmetry bounds in the mean field spin glass
  model,'' {\em Communications in {M}athematical {P}hysics}, vol.~233, no.~1,
  pp.~1--12, 2003.

\bibitem{ghirlanda1998general}
S.~Ghirlanda and F.~Guerra, ``General properties of overlap probability
  distributions in disordered spin systems. {T}owards {P}arisi
  ultrametricity,'' {\em Journal of {P}hysics {A}: {M}athematical and
  {G}eneral}, vol.~31, no.~46, p.~9149, 1998.

\bibitem{talagrand2000rsb}
M.~Talagrand, ``Replica symmetry breaking and exponential inequalities for the
  {S}herrington-{K}irkpatrick model,'' {\em The {A}nnals of {P}robability},
  vol.~28, no.~3, pp.~1018--1062, 2000.

\bibitem{panchenko2010ultra}
D.~Panchenko, ``A connection between the {G}hirlanda--{G}uerra identities and
  ultrametricity,'' {\em The {A}nnals of {P}robability}, vol.~38, no.~1,
  pp.~327--347, 2010.

\bibitem{panchenko2011ultra}
D.~Panchenko, ``Ghirlanda--{G}uerra identities and ultrametricity: An
  elementary proof in the discrete case,'' {\em Comptes {R}endus
  {M}athematique}, vol.~349, no.~13-14, pp.~813--816, 2011.

\bibitem{panchenko2013parisi}
D.~Panchenko, ``The parisi ultrametricity conjecture,'' {\em Annals of
  {M}athematics}, pp.~383--393, 2013.

\bibitem{amit1985}
D.~J. Amit, H.~Gutfreund, and H.~Sompolinsky, ``Storing infinite numbers of
  patterns in a spin-glass model of neural networks,'' {\em Physical {R}eview
  {L}etters}, vol.~55, no.~14, p.~1530, 1985.

\bibitem{hopfield1982hopfield}
J.~J. Hopfield, ``Neural networks and physical systems with emergent collective
  computational abilities,'' {\em Proceedings of the national academy of
  sciences}, vol.~79, no.~8, pp.~2554--2558, 1982.

\bibitem{pastur1977exactly}
L.~A. Pastur and A.~L. Figotin, ``Exactly soluble model of a spin glass,'' {\em
  Journal of {L}ow {T}emperature {P}hysics}, vol.~3, no.~6, pp.~378--383, 1977.

\bibitem{Krotov2016DenseAM}
D.~Krotov and J.~J. Hopfield, ``Dense associative memory for pattern
  recognition,'' {\em Advances in neural information processing systems},
  vol.~29, 2016.

\bibitem{Krotov2018DAMS}
D.~Krotov and J.~Hopfield, ``Dense associative memory is robust to adversarial
  inputs,'' {\em Neural computation}, vol.~30, no.~12, pp.~3151--3167, 2018.

\bibitem{AD-EPJ2020}
E.~Agliari and G.~De~Marzo, ``Tolerance versus synaptic noise in dense
  associative memories,'' {\em The {E}uropean {P}hysical {J}ournal {P}lus},
  vol.~135, no.~11, pp.~1--22, 2020.

\bibitem{baldi1987number}
P.~Baldi and S.~S. Venkatesh, ``Number of stable points for spin-glasses and
  neural networks of higher orders,'' {\em Physical {R}eview {L}etters},
  vol.~58, no.~9, p.~913, 1987.

\bibitem{bovier2001spin}
A.~Bovier and B.~Niederhauser, ``The spin-glass phase-transition in the
  {H}opfield model with $p$-spin interactions,'' {\em Mathematical {P}hysics
  and {M}athematics}, 2001.

\bibitem{steffan1994replica}
H.~Steffan and R.~K{\"u}hn, ``Replica symmetry breaking in attractor neural
  network models,'' {\em Zeitschrift f{\"u}r {P}hysik {B} {C}ondensed
  {M}atter}, vol.~95, no.~2, pp.~249--260, 1994.

\bibitem{tanaka2007moment}
T.~Tanaka, ``Moment problem in replica method,'' {\em Interdisciplinary
  information sciences}, vol.~13, no.~1, pp.~17--23, 2007.

\bibitem{guerra2001sum}
F.~Guerra, ``Sum rules for the free energy in the mean field spin glass
  model,'' {\em Fields Institute {C}ommunications}, vol.~30, no.~11, 2001.

\bibitem{agliari2012notes}
E.~Agliari, A.~Barra, R.~Burioni, and A.~Di~Biasio, ``Notes on the p-spin glass
  studied via {H}amilton-{J}acobi and smooth-cavity techniques,'' {\em Journal
  of {M}athematical {P}hysics}, vol.~53, no.~6, p.~063304, 2012.

\bibitem{barra2013mean}
A.~Barra, G.~Dal~Ferraro, and D.~Tantari, ``Mean field spin glasses treated
  with {PDE} techniques,'' {\em The {E}uropean {P}hysical {J}ournal {B}},
  vol.~86, no.~7, pp.~1--10, 2013.

\bibitem{barra2008mean}
A.~Barra, ``The mean field ising model trough interpolating techniques,'' {\em
  Journal of {S}tatistical {P}hysics}, vol.~132, no.~5, pp.~787--809, 2008.

\bibitem{barra2010replica}
A.~Barra, A.~Di~Biasio, and F.~Guerra, ``Replica symmetry breaking in
  mean-field spin glasses through the {Hamilton}--{Jacobi} technique,'' {\em
  Journal of {S}tatistical {M}echanics: {T}heory and {E}xperiment}, vol.~2010,
  no.~09, p.~P09006, 2010.

\bibitem{barra2014proc}
A.~Barra, A.~Di~Lorenzo, F.~Guerra, and A.~Moro, ``On quantum and relativistic
  mechanical analogues in mean-field spin models,'' {\em Proceedings of the
  {R}oyal {S}ociety {A}: {M}athematical, {P}hysical and {E}ngineering
  {S}ciences}, vol.~470, no.~2172, p.~20140589, 2014.

\bibitem{Moro2018annals}
G.~De~Matteis, F.~Giglio, and A.~Moro, ``Exact equations of state for
  nematics,'' {\em Annals of {P}hysics}, vol.~396, pp.~386--396, 2018.

\bibitem{AMT-JSP2018}
E.~Agliari, D.~Migliozzi, and D.~Tantari, ``Non-convex multi-species {H}opfield
  models,'' {\em Journal of {S}tatistical {P}hysics}, vol.~172, no.~5,
  pp.~1247--1269, 2018.

\bibitem{Moro2019PRE}
P.~Lorenzoni and A.~Moro, ``Exact analysis of phase transitions in mean-field
  {P}otts models,'' {\em Physical {R}eview {E}}, vol.~100, no.~2, p.~022103,
  2019.

\bibitem{ABN-JMP2019}
E.~Agliari, A.~Barra, and M.~Notarnicola, ``The relativistic {H}opfield
  network: {R}igorous results,'' {\em Journal of {M}athematical {P}hysics},
  vol.~60, no.~3, p.~033302, 2019.

\bibitem{fachechi2021pde}
A.~Fachechi, ``{PDE}/statistical mechanics duality: Relation between {G}uerra's
  interpolated $p$-spin ferromagnets and the {B}urgers hierarchy,'' {\em
  Journal of {S}tatistical {P}hysics}, vol.~183, no.~1, pp.~1--28, 2021.

\bibitem{AAAF-JPA2021}
E.~Agliari, L.~Albanese, F.~Alemanno, and A.~Fachechi, ``A transport equation
  approach for deep neural networks with quenched random weights,'' {\em
  Journal of {P}hysics {A}: {M}athematical and {T}heoretical}, vol.~54, no.~50,
  p.~505004, 2021.

\bibitem{Barra2015Annals}
A.~Barra and A.~Moro, ``Exact solution of the {V}an der {W}aals model in the
  critical region,'' {\em Annals of {P}hysics}, vol.~359, pp.~290--299, 2015.

\bibitem{DeNittis2012PrsA}
G.~De~Nittis and A.~Moro, ``Thermodynamic phase transitions and shock
  singularities,'' {\em Proceedings of the {R}oyal {S}ociety {A}:
  {M}athematical, {P}hysical and {E}ngineering {S}ciences}, vol.~468, no.~2139,
  pp.~701--719, 2012.

\bibitem{Giglio2016Physica}
F.~Giglio, G.~Landolfi, and A.~Moro, ``Integrable extended {V}an der {W}aals
  model,'' {\em Physica {D}: {N}onlinear {P}henomena}, vol.~333, pp.~293--300,
  2016.

\bibitem{Moro2014Annals}
A.~Moro, ``Shock dynamics of phase diagrams,'' {\em Annals of {P}hysics},
  vol.~343, pp.~49--60, 2014.

\bibitem{derrida1981rem}
B.~Derrida, ``Random-energy model: {A}n exactly solvable model of disordered
  systems,'' {\em Physical {R}eview {B}}, vol.~24, no.~5, p.~2613, 1981.

\bibitem{ShcherbinaPastur-JSP1991}
L.~Pastur and M.~Shcherbina, ``Absence of self-averaging of the order parameter
  in the {S}herrington-{K}irkpatrick model,'' {\em Journal of {S}tatistical
  {P}hysics}, vol.~62, no.~1, pp.~1--19, 1991.

\bibitem{Bovier-JPA1994}
A.~Bovier, ``Self-averaging in a class of generalized {H}opfield models,'' {\em
  Journal of {P}hysics {A}: {M}athematical and {G}eneral}, vol.~27, no.~21,
  p.~7069, 1994.

\bibitem{amit1989}
D.~J. Amit, {\em Modeling brain function: {T}he world of attractor neural
  networks}.
\newblock Cambridge university press, 1989.

\bibitem{SpecialIssue-JPA}
E.~Agliari, A.~Barra, P.~Sollich, and L.~Zdeborov{\'a}, ``Machine learning and
  statistical physics: preface,'' {\em Journal of {P}hysics {A}: {M}athematical
  and {T}heoretical}, vol.~53, p.~500401, 2020.

\bibitem{AABO-JPA2020}
E.~Agliari, L.~Albanese, A.~Barra, and G.~Ottaviani, ``Replica symmetry
  breaking in neural networks: a few steps toward rigorous results,'' {\em
  Journal of {P}hysics {A}: {M}athematical and {T}heoretical}, vol.~53, no.~41,
  p.~415005, 2020.

\bibitem{albanese2021}
L.~Albanese, F.~Alemanno, A.~Alessandrelli, and A.~Barra, ``Replica symmetry
  breaking in dense neural networks,'' {\em arXiv preprint arXiv:2111.12997},
  2021.

\bibitem{gardner1987}
E.~Gardner, ``Multiconnected neural network models,'' {\em Journal of {P}hysics
  {A}: {M}athematical and {G}eneral}, vol.~20, no.~11, p.~3453, 1987.

\bibitem{AABCF-PRL2020}
E.~Agliari, F.~Alemanno, A.~Barra, M.~Centonze, and A.~Fachechi, ``Neural
  networks with a redundant representation: detecting the undetectable,'' {\em
  Physical {R}eview {L}etters}, vol.~124, no.~2, p.~028301, 2020.

\bibitem{AABF-NN2020}
E.~Agliari, F.~Alemanno, A.~Barra, and A.~Fachechi, ``Generalized {G}uerra's
  interpolation schemes for dense associative neural networks,'' {\em Neural
  {N}etworks}, vol.~128, pp.~254--267, 2020.

\bibitem{olver1977evolution}
P.~J. Olver, ``Evolution equations possessing infinitely many symmetries,''
  {\em Journal of {M}athematical {P}hysics}, vol.~18, no.~6, pp.~1212--1215,
  1977.

\bibitem{tasso1976cole}
H.~Tasso, ``{Cole's ansatz and extensions of {B}urgers' equation},'' tech.
  rep., Max-{P}lanck-Institut f{\"u}r {P}lasmaphysik, 1976.

\bibitem{Genovese_2012}
G.~Genovese, ``Universality in bipartite mean field spin glasses,'' {\em
  Journal of {M}athematical {P}hysics}, vol.~53, no.~12, p.~123304, 2012.

\bibitem{Agliari_2019}
E.~Agliari, A.~Barra, and B.~Tirozzi, ``Free energies of {B}oltzmann machines:
  self-averaging, annealed and replica symmetric approximations in the
  thermodynamic limit,'' {\em Journal of {S}tatistical {M}echanics: {T}heory
  and {E}xperiment}, vol.~2019, no.~3, p.~033301, 2019.

\end{thebibliography}

\end{document}